\iftrue
\documentclass[a4paper,UKenglish]{lipics-v2016}
\fi
\usepackage{enumerate}
\usepackage{amsthm,amsmath,amsfonts}
\usepackage{todonotes}

\usepackage{makecell}

\newcommand{\be}{\begin{equation}}
\def\eeeForColoration\end{equation}
\newcommand{\ee}{\end{equation}}

\DeclareMathOperator{\last}{last}
\DeclareMathOperator{\first}{first}

\newcommand{\nats}{\mathbb{N}}
\newcommand{\NN}{\nats}

\newcommand{\ov}{\widehat}
\newcommand{\set}[1]{\{#1\}}

\newcommand{\pfun}{\rightharpoonup}
\DeclareMathOperator{\dom}{dom}

\DeclareMathOperator{\ima}{im}
\newcommand{\lab}{t}

\newcommand{\Proba}{\mathbb{P}}

\newcommand{\aPath}{w}
\newcommand{\states}{S}
\newcommand{\pt}{p}
\newcommand{\distrib}{\Delta}
\newcommand{\mc}{\mathcal{M}}

\newcommand{\proba}{\Proba}
\DeclareMathOperator{\pathes}{Path}

\newcommand{\npconp}{{\sc np}$ \cap ${\normalfont co}{\sc np}}
\newcommand{\exptime}{{\sc exptime}}
\newcommand{\nexptime}{{\sc nexptime}}
\newcommand{\ptime}{{\sc ptime}}

\newcommand{\tree}{t}

\newcommand{\stateFormula}{\psi}
\newcommand{\pathFormula}{\phi}
\newcommand{\pf}{\pathFormula}
\renewcommand{\sf}{\stateFormula}
\newcommand{\pCTLFormula}{\xi}

\newcommand{\fo}{\pCTLFormula}

\newcommand{\PCTLE}{\exists}
\newcommand{\PCTLA}{\forall}
\newcommand{\PCTLPp}{\proba_{>0}}
\newcommand{\PCTLPo}{\proba_{=1}}

\renewcommand{\next}{X}
\newcommand{\until}{U}
\newcommand{\always}{G}

\newcommand{\ctl}{{CTL}}
\newcommand{\ctls}{{\ctl}$^*$}
\newcommand{\pctls}{p\ctls}
\newcommand{\ECTL}{ECTL}
\newcommand{\allop}{$[\PCTLE,\PCTLA,\PCTLPp,\PCTLPo]$}
\newcommand{\pop}{$[\PCTLPp,\PCTLPo]$}
\newcommand{\ECTLa}{\ECTL\allop}
\newcommand{\ctlsa}{\ctls\allop}

\newcommand{\automata}{\mathcal{A}}
\newcommand{\Astates}{Q}

\newcommand{\Atrans}{\to}

\newcommand{\Q}{\Astates}
\newcommand{\QA}{\Astates_A}
\newcommand{\QE}{\Astates_E}
\newcommand{\QN}{\Astates_N}

\newcommand{\Qall}{F_\forall}

\newcommand{\Qone}{F_1}
\newcommand{\Qzero}{F_{>0}}

\newcommand{\ignore}{\sharp}

\newcommand{\lang}{\mathcal{L}}

\newcommand{\ndautomata}{\mathcal{B}}

\newcommand{\blank}{\circ}
\newcommand{\BSP}{{\bf BIN-SAT}}
\newcommand{\SP}{{\bf MC-SAT}}

\DeclareMathOperator{\inde}{index}
\DeclareMathOperator{\lar}{LAR}
\renewcommand{\index}{\inde}

\newcommand{\limch}{limited choice}

\newcommand{\what}{\limch\ for Adam}

\usepackage{mathtools}

\begin{document}

\iftrue
\bibliographystyle{plainurl}

\title{Alternating Nonzero Automata}

\author[1]{Paulin Fournier}
\author[1]{Hugo Gimbert}
\affil[1]{LaBRI, CNRS, Universit{\' e} de Bordeaux, France.\\
\texttt{\{paulin.fournier,hugo.gimbert\}@labri.fr}}


\maketitle

\fi

\begin{abstract}
We introduce 
a new class of automata on infinite trees
called
\emph{alternating nonzero automata},
which extends the class of non-deterministic nonzero automata.
We reduce the emptiness problem
for alternating nonzero automata
to the same problem for non-deterministic ones, which implies decidability.
We obtain as a corollary algorithms for the satisfiability of a probabilistic 
temporal 
logic extending both CTL* and the qualitative fragment of pCTL*.
\end{abstract}

\iftrue
\section{Introduction}
\fi

The theory of automata on infinite trees is
rooted in Rabin's seminal theorem which establishes an 
effective correspondence between the monadic second order logic (MSO)
theory of the infinite binary tree and the non-deterministic automata on this tree~\cite{rabinsem}. In this correspondence, the satisfiability of the logic is dual to the emptiness of the algorithm and both these algorithmic problems are mutually reducible to one another.

This elegant setting has been partially extended to probabilistic logics~\cite{LS1982,brazdil2008controller,DBLP:conf/lfcs/MichalewskiM16,DBLP:journals/corr/MichalewskiMB16,DBLP:conf/icalp/Bojanczyk16}
and automata with probabilistic winning conditions
~\cite{rabinsem, pazbook,DBLP:journals/jacm/BaierGB12,DBLP:journals/tocl/CarayolHS14,DBLP:conf/icalp/Bojanczyk16}. In this paper we make another step in this direction: we show a correspondence between the logic \ctlsa\ 
and nonzero alternating automata with limited choice. Moreover we show  that the emptiness problem of the automata is decidable and obtain as a corollary the decidability of the satisfiability of the logic.

\paragraph*{Automata.}
Alternating nonzero automata
are an alternating version of \emph{non-deterministic nonzero automata} introduced in~\cite{DBLP:journals/corr/BojanczykGK17},
which themselves are equivalent to \emph{non-deterministic zero automata} introduced in~\cite{DBLP:conf/icalp/Bojanczyk16}.

An alternating nonzero automaton  
takes as input a binary tree. 
Some states of the automaton are controlled by Eve, while other states are controlled by Adam, and the player controlling the current state chooses the next transition. Some transitions are \emph{local transitions}, in which case the automaton stays on the same node of the input tree while other are \emph{split transitions} in which case the automaton proceeds to the left son or to the right son of the current node with equal probability $\frac{1}{2}$.

This interaction between Eve and Adam is seen as a game where Eve and Adam play according to some strategies. Once the strategies are fixed, one obtains a Markov chain whose trajectories are all possible plays consistent with the strategies. The winner is determined with respect to winning conditions introduced in~\cite{DBLP:conf/icalp/Bojanczyk16,DBLP:journals/corr/BojanczykGK17}, using
a total order on the set of states (used to compute the limsup of a play which is the largest state seen infinitely often during the play)
and three subsets of states,
respectively called the \emph{sure}, \emph{almost-sure} and \emph{positive states}.
Eve wins if and only if the three acceptance  conditions hold:

{\noindent \bf sure winning:} every play has limsup in sure states; and 

{\noindent \bf almost-sure winning:} almost-every play has limsup in almost-sure states; and 

{\noindent \bf positive winning:} whenever the play enters a positive state there is positive probability that the play never exits positive states.

The input tree is accepted by the alternating automaton iff Eve has a winning strategy.

Alternating nonzero automata generalize both classical alternating automata with parity conditions~\cite{Chandra:1981:ALT:322234.322243, MULLER1987267} (when all states are almost-sure and positive) 
as well as non-determin\-istic nonzero automata~\cite{DBLP:journals/corr/BojanczykGK17} (in case Eve controls all states).

We do not know whether the emptiness problem for these automata is decidable or not, however we show that the answer is positive for the subclass of
alternating nonzero automata with \emph{\what}.
In these automata, some choices of Adam are canonical, at most one in every state, and Adam may perform at most a bounded number of 
non-canonical choices during a single play.

We establish some properties of alternating nonzero automata with \what.

\begin{itemize}
\item
First, we show that the emptiness problem for alternating nonzero automata with \what\ is in 
\nexptime $\cap$ co-\nexptime\ (Theorem~\ref{theo:recogalt}). 
The proof is an \exptime\ reduction to the emptiness problem for non-deterministic automata. This proof
relies on the positional determinacy of the acceptance games for Eve (Lemma~\ref{lem:det})
and a characterization of positional winning strategies for Eve (Lemmas~\ref{caracsure},~\ref{defindex} and~\ref{lem:caracnonzero}).
\item
Second, we show that in the particular case where the sure winning condition 
is a B\"uchi condition,
emptiness of non-deterministic nonzero automata
is in \ptime\ (Theorem~\ref{theo:complexitiyemptiness})
hence, in case of 
a B\"uchi sure winning condition, 
emptiness of nonalternating nonzero automata is in \exptime\  (Theorem~\ref{theo:recogalt}).
\end{itemize}

\paragraph*{Logic.}
The temporal logic \ctls\ introduced by Emerson and Halpern \cite{emerson1986sometimes} and its fragments \ctl\ and LTL are    prominent tools to specify properties of discrete event systems.

%

A variant of \ctls\ is the logic
\pctls\ \cite{hansson1994logic} in which the universal and existential path quantifiers are \emph{replaced} by probabilistic path quantifiers which set upper or lower bounds on the probability of a path property in a Markov chain.
For example the formula $\proba_{\geq \frac{1}{2}}(FGa)$ specify that
with probability at least $\frac{1}{2}$ eventually all the visited states are labelled with $a$.
To our knowledge, the satisfiability problem for this logic is an open problem. 

However, for
the qualitative fragment of \pctls, where only two probabilistic quantifiers $\proba_{>0}$ and $\proba_{=1}$ are available, the satisfiability
is decidable~\cite{brazdil2008controller}. In a variant of \pctls\ called p\ECTL\ the path subformula are replaced by deterministic B\"uchi automaton, and the satisfiability of the qualitative fragment is 2-\exptime\ complete~\cite{brazdil2008controller}, the same complexity as for \ctls~\cite{vardi1985improved}.

Remark that neither \pctls\ nor  p\ECTL\ 
includes the path operators $\forall$ and $\exists$,
thus  these two logics are incomparable in expressivity with \ctls.
For example, on the alphabet $\{a,b\}$, the \ctls\ formula $\phi_1=\forall F G \neg b$, and the \pctls\ formula $\phi_2=\PCTLPo( F G \neg b)$ specify, that \emph{every} branch, respectively \emph{almost-every} branch, of the model has finitely many $b$.
Neither $\phi_1$ can be expressed in p\ctls\ nor $\phi_2$
can be expressed in \ctls.

\smallskip 

In this paper, we consider the logic \ctlsa\ which is an extension of both \ctls\ and qualitative \pctls
and establish several properties of this logic.

\begin{itemize}
\item
The satisfiability by an arbitrary $\Sigma$-labelled Markov chain reduces 
to the satisfiability by 
$(\Sigma\cup \{\circ\})$-labelled 
a binary tree
with $\blank$ a fresh letter
(Theorem~\ref{theo:reduc}).
\item
The satisfiability of \ctlsa\ reduces to the emptiness
of alternating nonzero automata with finite choice for Adam
thus it is decidable in 3-\nexptime $\cap$co-3-\nexptime.
In the variant \ECTLa, where path formula are deterministic B\"uchi automata, this reduction gives a
2-\nexptime $\cap$ co-2-\nexptime\ complexity
and for the fragment \ctl\allop\ the complexity is \nexptime $\cap$ co-\nexptime\ (Theorem~\ref{theo:pctls}).
\item
For the fragments \ctls\pop,
\ECTL\pop\
and \ctl\pop\
(i.e. qualitative p\ctls, p\ECTL\ and p\ctl{} respectively),
the $\Qall$ acceptance condition of the automaton is a B\"uchi condition and we retrieve the optimal complexity bounds of~\cite{brazdil2008controller,Brazdil2008}, i.e. 3-\exptime, 2-\exptime\
and \exptime, respectively.
\end{itemize}


\paragraph*{Organization of the paper.} 
Section~\ref{sec:nonzero} introduces alternating nonzero automata,
an example is given in Section~\ref{sec:example}.
Section~\ref{sec:nondet} focuses on non-deterministic automata,
and provide an optimal algorithm to decide emptiness
when the $\Qall$ condition is B\"uchi.
In Section~\ref{sec:emptiness} 
we prove that emptiness is decidable (\ref{theo:recogalt})
when Adam has limited choice.
Section~\ref{sec:pctls} presents our complexity results
for the satisfiability of \ctlsa\ and its variants and fragments.
\section{Alternating nonzero automata}\label{sec:nonzero}


An alternating nonzero automaton on a finite alphabet $\Sigma$
is a finite-state machine processing binary trees,
equipped with a game semantics:
every tree is either accepted or rejected by the machine
depending on who wins the acceptance game on the tree.

\paragraph*{Trees.} A $\Sigma$-labelled binary tree is a function $\tree:\{0,1\}^*\to\Sigma$.
An element $n\in\{0,1\}^*$ is called a \emph{node} of the tree
and has exactly two sons
$n0$ and $n1$.
 We use the usual notions of ancestors and descendants.
A node $n'$ is \emph{(strictly) below} $n$ if $n$ is a (strict) prefix of $n'$.
A \emph{path} in the tree is a finite or infinite sequence of nodes $n_0,n_1,\ldots$
such that for every $k$ the node $n_{k+1}$ is a son of the node $n_k$.

A branch $b$ is an element of $\{0,1\}^\omega$.
If a node $n$ is a prefix of $b$ we say that $n$ \emph{belongs} to $b$ or that $b$ \emph{visits} $n$.
The set of branches is equipped with the uniform probability measure, denoted $\mu$,
corresponding to an infinite random walk taking at each step either direction $0$ or $1$ with equal probability $\frac{1}{2}$.

\paragraph*{Automata.}
An  alternating nonzero automaton on alphabet $\Sigma$
is presented as a tuple
\[
\automata=(\Q,q_0,\QE,\QA,\Atrans,
\Qall, \Qone,\Qzero
)\text{ where:}
\]
\begin{itemize}
  \item $\Astates$ is a finite set of states, equipped with a total order $\leq$, containing the initial state $q_0$.
  \item $(\QE,\QA)$ is a partition of $\Q$ into Eve and  Adam states.
  \item $\to$ is the set of transitions,
  there are two types of transitions:
\begin{itemize}
\item  \emph{local transitions} are
tuples $(q,a,q')$ with $q,q'\in\Q$ and $a\in \Sigma$, denoted $q\to_a q'$. 
\item  \emph{split transitions} are tuples $(q,a,q_0,q_1)\in \Q\times \Sigma\times Q^2$, denoted $q\to_a (q_0,q_1)$.
  \end{itemize}
  \item
  $\Qall$, $\Qone$ and $\Qzero$
  are subsets of $Q$
  defining the acceptance condition. 
    \end{itemize}

The input of such an automaton is an infinite binary tree $t : \{0,1\}^* \to \Sigma$. The source (resp. the target) of a local transition $q\to_a q'$ is $q$ (resp $q'$).
The source (resp. the targets) of a split transition $q \to_a (q_0,q_1)$ is $q$ (resp $q_0$ and $q_1$).
A state is said to be controlled by Eve or Adam 
whether it belongs to $\QE$ or $\QA$. 
The controller of a transition is the controller of its source state.
We always assume that
    \begin{itemize}
    \item[{\bf (HC)}]
    the automaton is {\bf complete}: for every state $q$ and letter $a$ there is at least one transition with source $q$ on $a$.
    \end{itemize}

The (HC) condition makes it easier to define the game semantics of the automaton.

\paragraph*{Game semantics.}
%
%
%

The acceptance of an input binary tree by the automaton is defined by mean of
a stochastic game between Eve and Adam called the \emph{acceptance game}.

The game of acceptance of a binary tree $\lab:\{0,1\}^*\to \Sigma$ by $\automata$
is a two-player stochastic game with perfect information
played by two strategic players Eve and Adam.
The vertices of the game are all pairs $(n,q)$ where $n\in\{0,1\}^*$ is
a node of the infinite binary tree and
$q$ is a state of the automaton.
The game starts in the initial vertex $(\epsilon,q_0)$.

Each vertex $(n,q)$ is controlled by either Eve or Adam depending whether $q\in \QE$ or $q \in \QA$.
The controller of the current state
%
chooses any transition with source $q$ and letter $\lab(n)$.
Intuitively, depending whether the transition is a local or a split transition,
the automaton stays on the current node $n$ or move with equal probability
$\frac{1}{2}$ to either node $n0$ or $n1$.
If the transition is a local transition $q\to_{\lab(n)} q'$,
the new vertex of the game
is $(n,q')$.
If the transition is a split transition $q \to_{\lab(n)} (r_0,r_1)$ then
the new vertex is chosen randomly with equal probability $\frac{1}{2}$
between vertices $(n0,r_0)$ or $(n1,r_1)$.
%

A play is a finite or infinite sequence of vertices $\pi=(n_0,q_0)(n_1,q_1)\ldots $.
We denote $\first(\pi) = (n_0,q_0)$
 and $\last(\pi) = (n_k,q_n)$ (for finite plays).

 A strategy for Eve associates with every finite play
whose last vertex is controlled by Eve
a transition with source $q_n$ and letter $\lab(n_k)$
(such a transition always exists since the automaton is complete).
Strategies for Adam are defined in a symmetric way.
Strategies of Eve are usually denoted $\sigma$ while strategies for Adam are denoted $\tau$.

\paragraph*{Measuring probabilities.}
Once both players Eve and Adam have chosen some strategies
$\sigma$ and $\tau$, this defines naturally
a non-homogenous Markov chain whose states
are the vertices of the game.
According to Tulcea theorem,
if we equip the set of plays with the $\sigma$-field generated by cylinders,
then there is a unique probability measure $\Proba^{\sigma,\tau}$
such that 
after a play
$\pi=(n_0,q_0)\ldots (n_k,q_k)$,
if $\delta(\pi)$ denotes the transition chosen by Eve or Adam
after $\pi$ (depending whether $q_k \in \QE$ or $q_k \in \QA$),
the probability to go to vertex $(n_{k+1},q_{k+1})$
is:
\[
\begin{cases}
 1 & \text{ if $\delta(\pi)$ is the local transition $q_k\to_{\lab(n_k)} q_{k+1}$}\enspace,\\
 \frac{1}{2}
& \text{ if $\delta(\pi)$ is the split transition $q_k\to_{\lab(n_k)} (r_0,r_1)$ and }\\
&\hspace{2cm}\begin{cases}
\text{$n_{k+1}=n_k0$ and $q_{k+1}= r_0$}\enspace; or\\
\text{$n_{k+1}=n_k1$ and $q_{k+1}=r_1$}\enspace.
\end{cases}\\
 0 & \text{ otherwise\enspace.}
\end{cases}
\]
This way we obtain a probability measure $\Proba^{\sigma,\tau}$
on the set of infinite plays.

\paragraph*{Consistency and reachability.}
If a finite play $\pi$ is the prefix of another finite or infinite play $\pi'$
we say that $\pi'$ is a \emph{continuation} of $\pi$.
A finite $\pi$ play is \emph{consistent}
with a strategy $\sigma$ or, more simply, is a \emph{$\sigma$-play}
if there exists a strategy $\tau$ such that $\pi$ may occur 
in the non-homogenous Markov chain induced by $\sigma$ and $\tau$.
In this case, the number $N$ of split transitions which occurred 
in $\pi$ is exactly the depth of the node of $\last(\pi)$
and
\[
\Proba^{\sigma,\tau}(\{ \text{ continuations of $\pi$ }\}) = 2^{-N}\enspace.
\]
A vertex $w$ is \emph{$\sigma$-reachable} if there exists a finite $\sigma$-play from the initial vertex to $w$. 
 An infinite play is consistent with $\sigma$ if all its prefixes are.

\paragraph*{Bounded vs. unbounded plays.}
There are two kinds of infinite plays: \emph{bounded plays} are plays whose sequence of nodes is ultimately constant, or equivalently which ultimately use only local transitions while \emph{unbounded plays}
use infinitely many split transitions.

Bounded plays consistent  with $\sigma$ and $\tau$
are the atoms of $\Proba^{\sigma,\tau}$:
a play $\pi$ is bounded and consistent
with $\sigma$ and $\tau$ iff $\Proba^{\sigma,\tau}(\{\pi\})>0$.

In this paper we will focus on subclasses of automata
whose structural restrictions forbids the existence of bounded plays
(see the {\bf (NLL)} hypothesis below).

So in practice, every play $\pi=(n_0,q_0)(n_1,q_1)\ldots$ we consider
will visit a sequence of nodes $n_0,n_1,n_2,\ldots$
which enumerates all finite prefixes of an infinite branch $b\in\{0,1\}^\omega$ of the binary tree,
in a weakly increasing order: for every index $i$ either $n_{i+1}=n_{i}$
(the player controlling $(n_i,q_i)$ played a local transition) or $n_{i+1}=n_i d$ for some $d\in \{0,1\}$
(the player controlling $(n_i,q_i)$ played a split transition and the play followed direction $d$). 

\paragraph*{Winning strategies.} Whether Eve wins the game
 is defined as follows.
The \emph{limsup} of an infinite play $(n_0,q_0)(n_1,q_1)\ldots$
is $\limsup_i q_i$ i.e. the largest automaton state visited infinitely often.
An infinite play $\pi'$ is a \emph{positive continuation} of $\pi$
if all states of $\pi'$ visited after $\pi$ belongs to $\Qzero$.

Eve wins with $\sigma$ against $\tau$ if the three following conditions are satisfied.
\begin{itemize}
\item {\bf Sure winning.} Every play consistent with $\sigma$ and $\tau$ has limsup in $\Qall$.
\item {\bf Almost-sure winning.} Almost-every play consistent with $\sigma$ and $\tau$ has limsup in $\Qone$.
\item{\bf Positive winning.}
For every finite play 
$\pi$
consistent with $\sigma$ and $\tau$
whose last state belongs to $\Qzero$,
the set of positive continuations of $\pi$
 has nonzero probability.
\end{itemize}

We say that \emph{Eve wins} the acceptance game
if she has a \emph{winning strategy} i.e. a strategy
which wins 
the acceptance game against any strategy of Adam.

%
%
%
%

%

\paragraph*{B\"uchi conditions.}
A B\"uchi condition is a set of states $R\subseteq Q$ which is upper-closed with respect to $\leq$\enspace. Then a play has limsup in $R$ iff it visits $R$ infinitely often.

\paragraph*{Language of an automaton.}
\begin{definition}[Acception and language]
A binary tree is \emph{accepted} by the automaton if Eve
has a winning strategy in the acceptance game.
The language of the automaton is the set of its accepted trees.
\end{definition}

We are interested in the following decision problem:

\medskip 
{\bf Emptiness problem: }
Given an automaton, decide whether its language is empty or not.
\medskip 

The use of game semantics makes the following closure properties trivial.
\begin{lemma}[Closure properties]\label{lem:closure}
The class of languages recognized by alternating nonzero automata
is closed under union and intersection.
\end{lemma}

\paragraph*{Normalization.}
We assume all automata to be normalized
in the sense where they satisfy:
\begin{itemize}
\item {\bf (N1)}
  every split transition whose source is in $\Qzero$ has at least one successor in $\Qzero$; and
  \item {\bf (N2)}
every local transition whose source is in $\Qzero$ has its target in $\Qzero$ as well.
\end{itemize}

We can normalize an arbitrary automaton
by removing 
all transitions violating  {\bf (N1)} and  {\bf (N2)}.
This will not change the language
because such transitions
 are never used by positively winning strategies of Eve.
 This normalization could lead to a violation of the completeness hypothesis,
 {\bf (HC)}. In this case we can also delete the corresponding states without modifying the language of the automaton.

If one would drop  {\bf (HC)}  then the game graph may have dead-ends and the rules of the game would have to be extended to handle this case, typically the player controlling the state in the dead-end loses the game. This extension does not bring any extra expressiveness to our model of automaton, we can always make an automaton complete by adding local transitions leading to losing absorbing states.

%
Moreover,
we assume:
\begin{itemize}
\item {\bf (N3)}
$\Qone \subseteq \Qall\enspace.$
\end{itemize}
This is w.l.o.g. since replacing $\Qone$ 
with $\Qone \cap \Qall$ does not modify the language of the automaton.

%




\section{An example: the language of PUCE trees\label{sec:example}}
A tree $\lab$ on the alphabet $\set{a,b}$ is  \emph{positively ultimately constant everywhere}
(\emph{PUCE} for short) if
for
every node $n$,
\begin{enumerate}
\item[i)] the set of branches visiting $n$ and with finitely many $a$-nodes has $>0$ probability; and
\item[ii)] the set of branches visiting $n$ and with finitely many $b$-nodes has $>0$ probability.
\end{enumerate}

\paragraph*{No regular tree is PUCE.}
There are two cases. If the regular tree has a node $n$ which is the root of a subtree labelled with only $a$ or $b$
then clearly the tree is not PUCE.
Otherwise, by a standard pumping argument,
every node labelled $a$ (resp. $b$) has a descendant labelled $b$ (resp. $a$)
at some depth $\leq |S|$, where $S$ is the set of states of the regular tree.
But in this second case from every node $n$ there is probability at least $\frac{1}{2^{|S|}}$
to reach a descendant with a different label, thus almost-every branch of the regular tree has infinitely many $a$ and $b$,
and the tree is not PUCE either.

\paragraph*{There exists a PUCE tree.}
However it is possible to build a non-regular tree $t$ whose every node satisfies both $i)$ and $ii)$. For that, we combine together two partial non-regular trees.
Let $H\subseteq \{0,1\}^*$ be a subset of nodes such that 
a) the set of branches which visit no node in $H$ has probability $\frac{1}{2}$,
b) no node of $H$ is a strict ancestor of another node in $H$ ($H$ is a cut), 
c) every node in $\{0,1\}^*$ is either a descendant or an ancestor of a node in $H$. For example we can choose 
$
H= \{00, 100, 0100, 11000, 011000, 1010000, 11100000, \\010100000,011100000,\ldots \}
$.

To obtain $t$, we combine two partial trees $t_a$ and $t_b$ whose domain
is $\{0,1\}^* \setminus H$ and $t_a$ is fully labeled with $a$ while $t_b$ is fully labelled with $b$.
Since $H$ is a cut, the nodes in $H$ are exactly the leaves of $t_a$ and $t_b$.
To obtain $t$, we plug a copy of $t_b$ on every leaf of $t_a$ and a copy of $t_a$ on every leaf of $t_b$.
Then from every node, according to c) there is non-zero probability to enter either $t_a$ or $t_b$
and according to a) there is non-zero probability to stay in there forever.

\paragraph*{An automaton recognizing PUCE trees.}
We can design one automaton for each of the two conditions
and combine them together with an extra state controlled by Adam
(cf proof of Lemma~\ref{lem:closure}).

We provide an alternating nonzero automaton checking condition ii), the automaton for condition i) is symmetric.
The state space is:
\[
Q = \set{s <  w < g  < \sharp}\enspace.
\]

Intuitively, Adam uses states $s$ to search for a node $n$ from which condition i) does not hold.
Once on node $n$, Adam switches to state $w$ and challenges Eve to find a path to an $a$-node $n'$ which is the root of an $a$-labelled subtree $T_n$ of $>0$ probability. For that Eve navigates the tree in state $w$ to node $n'$, switches to state $g$ on node $n'$, stays in $g$ as long as the play stays in $T_n$ and switches definitively to  $\sharp$
whenever leaving $T_n$. 

Formally, the only state controlled by Adam is $s$, i.e. $\QA=\{s\}$, from which Adam can choose, independently of the current letter, between two split transitions $s \to (s,\sharp)$ and $s\to (\sharp,s)$ and a local transition $s \to w$.
The state $\sharp$ is absorbing.
From state $w$, Eve can guess the path to $n'$ using the split transitions:
\[
w \to (\sharp, w) \quad w \to (w,\sharp)\enspace.
\]
Once $n'$ is reached Eve can switch to state $g$ with a local transition $w \to g$ and, whenever the current node is an $a$-node,
she can choose among  three split transitions:
\[
g \to_a (g,g) \quad
g \to_a (g,\sharp) \quad
g \to_a (\sharp,g) \enspace.
\]
The acceptance conditions are:
\begin{align*}
\Qall=\Qone=Q \setminus \{w\}
\quad
\quad
\Qzero= \set{ g }\enspace,
\end{align*}
so that from $w$ Eve is forced to eventually switch to $g$ (otherwise $\limsup=w\not\in\Qall$) and the $a$-subtree labelled by $g$ must have positive probability for Eve to win.
Adam may never exit the pathfinding state $s$, in which case Eve wins.

\section{Non-deterministic nonzero automata\label{sec:nondet}}
Non-deterministic \emph{zero} automata were introduced in~\cite{DBLP:conf/icalp/Bojanczyk16}, followed by a variant
of equivalent expressiveness,
non-deterministic \emph{nonzero} 
automata~\cite[Lemma 5]{DBLP:conf/icalp/BojanczykGK17}.
In those automata,  Adam is a dummy player,
i.e. 
$\QA=\emptyset$
and
moreover all transitions are split-transitions.

\begin{theorem}\label{theo:complexitiyemptiness}
The emptiness problem for non-deterministic nonzero automata is in \npconp.
If $\Qall$ is a B\"uchi condition then emptiness can be decided in \ptime.
\end{theorem}

The first statement is established in~\cite[Theorem 3]{DBLP:journals/corr/BojanczykGK17}.
The second statement is proved in the appendix.
The proof idea is as follows.
Assume the alphabet to be a singleton,
which is w.l.o.g. for non-deterministic automata.
The existence of a winning strategy for Eve
can be witnessed by a subset $W\subseteq Q$
which contains the initial state and 
two positional winning strategies $\sigma_1,\sigma_2: W \to W\times W$.
Strategy $\sigma_1$ should be almost-surely and positively winning
while strategy $\sigma_2$ should be surely winning. These two strategies  can be combined into a (non-positional) strategy for Eve which satisfies the three objectives, thus witnesses non-emptiness of the automaton.

\section{Deciding emptiness of automata with \what\label{sec:emptiness}}

In this section, we introduce the class of automata with \emph{\what},
and show that emptiness of these automata is decidable.

For that we rely on a characterization of positional strategies of Eve
which satisfy the surely and almost-surely winning conditions
(Lemma~\ref{caracsure},
Lemma~\ref{defindex})
and the positively winning condition (Lemma~\ref{lem:caracnonzero}).
Then we represent the positional strategies of Eve as labelled trees,
called \emph{strategic trees} (Definition~\ref{def:st}).
Finally we show that the language 
 of  strategic trees 
whose corresponding positional strategy is winning
can be recognized by a non-deterministic nonzero automaton (Theorem~\ref{theo:recogstrat}).
 
\subsection{Automata with \what}

In the rest of the paper, we focus on the class of automata with \what. Our motivation is that these automata capture the logic we are interested in and their acceptance games have good properties.
In particular the existence of positional winning strategies for Eve is one of the key properties used to decide emptiness.

To define the class of automata with \limch\ for Adam,
we rely on the transition graph of the automaton.
\begin{definition}[Equivalent and transient states]
The transitions of the automaton
define a directed graph called the \emph{transition graph} and denoted
$G_\to$. The vertices of $G_\to$ are 
$\Q$ and the edges are labelled with $\Sigma$, those are 
 all triplets $(q,a,r)$ such that
$q\to_a r$ is a local transition 
or such that
 $q\to_a(r,q')$ 
or
$q\to_a(q',r)$
is a split transition for some state $q'$.

Two states $q,r$ are \emph{equivalent}, denoted $q\equiv r$,
if they are in the same connected component
of $G_\to$.

A state is \emph{transient} if it does not belong to any connected component
of $G_\to$, or equivalently if there is no cycle on this state in $G_\to$.
\end{definition}

\begin{definition}
An automaton has \what\ if for every state $q$ controlled by Adam,
\begin{itemize}
\item
all transitions with source $q$ are local transitions; and
\item
for every letter $a$,
at most one of the (local) transitions $q\to_a q'$
satisfies $q \equiv q'$.
Such a transition is called a \emph{canonical} transition.
\end{itemize}
\end{definition}

In a \what\ automaton, the only freedom of choice of Adam, apart from playing canonical transitions, is deciding to go to a lower connected component of the transition graph. This non-canonical decision can be done only finitely many times,
hence the name \emph{limited choice}.

In the classical (non-probabilistic) theory of alternating automata,
similar notions of limited alternation have already been considered, for example \emph{hesitant alternating automata}~\cite{ltl}.

\begin{definition}[Canonical plays and transient vertices]
A \emph{canonical play} is a play in which Adam only plays canonical transitions.
A vertex $(n,q)$ of an acceptance game is \emph{transient} if it has no immediate successor $(n',q')$ (by a local or a split transition) such that $q \equiv q'$.
\end{definition}

In the acceptance game of an automaton with \what, every infinite play visit finitely many transient vertices and has a canonical suffix.

\paragraph*{The no local loop assumption.}
We assume that every automata with \what\ also satisfies:
\begin{itemize}
 \item {\bf (NLL)}
    the automaton has {\bf no local loop}: there is no letter $a$ and sequence of local transitions $q_0 \to_a q_1 \to_a \cdots \to_a q_i$
    such that $q_0=q_i$.
\end{itemize}
Under the hypothesis (NLL),
for every infinite play $\pi$ there is a unique branch of the binary tree
$b\in\{0,1\}^\omega$
 whose every prefix is visited 
by $\pi$. We say that $\pi$ \emph{projects} to $b$.

Assuming (NLL) does not reduce expressiveness.
\begin{lemma}\label{NLL}
Given an automaton $\automata$ with \what\ and set of 
states $Q$
one can effectively construct
another automaton $\automata'$ with \what\
satisfying {\bf (NLL)} and recognizing the same language.
\end{lemma}

The interest of the {\bf (NLL)} assumption is to make the acceptance game acyclic, which in turn guarantees positional determinacy for Eve, as shown in the next section.

The transformation performed in the proof of Lemma~\ref{NLL}
creates an exponential blowup of the state space of the automaton, which is bad for complexity.
We could do without this blowup by dropping the {\bf (NLL)} assumption, in which case Eve might need one extra bit of memory in order to implement local loops with priority in $\Qall\setminus\Qone$.

However, we prefer sticking to the {\bf (NLL)} assumption, which makes the alternating automata and their accepting games simpler and is anyway not restrictive when it comes to translating temporal logics into alternating automata: the natural translation produces automata with no local loop.

Another interest of the {\bf (NLL)} assumption is:
\begin{lemma}\label{probmu}
Assume the automaton has the (NLL) property.
Let $\mu$ be the uniform measure on the set of branches of the infinite binary tree, equipped with the usual Borel $\sigma$-field.
Let $t$ be an input tree, $\sigma$ and $\tau$ be two strategies 
in the corresponding acceptance game and $X$ be a measurable set of plays
consistent with $\sigma$ and $\tau$.
Let $Y\subseteq \{0,1\}^\omega$ be the set of infinite branches that $X$ projects to.
If $X$ is measurable then $Y$ is measurable and
\[
\Proba^{\sigma,\tau}(X) = \mu(Y) \enspace.
\]
\end{lemma}

\subsection{Positional determinacy of the acceptance game}

A crucial property of automata with \what\ is that their acceptance games are positionally determined for Eve.

\begin{definition}[Positional strategies]
A strategy $\sigma$ of Eve in an acceptance game
is \emph{positional} if for every finite plays $\pi,\pi'$ whose last vertices are controlled by $Eve$
and coincide, i.e. $\last(\pi)=\last(\pi')\in \{0,1\}^*\times \QE$,
then  $\sigma(\pi)=\sigma(\pi')$.
\end{definition}

\begin{lemma}[Positional determinacy for Eve]\label{lem:det}
Every acceptance game of an automaton with \what\
 is positionally determined for Eve: if Eve wins then she has a positional winning strategy.
\end{lemma}
\begin{proof}[Sketch of proof]
Since the (NLL) hypothesis is assumed, the underlying acceptance game is acyclic. The construction of a positional winning strategy $\sigma'$ from a (non-positional) winning strategy $\sigma$
relies on the selection of a canonical way of reaching a $\sigma$-reachable vertex $w$ 
with a $\sigma$-play $\pi(w)$ and setting $\sigma'(w)=\sigma(\pi(w))$. 
\end{proof}

\subsection{On winning positional strategies of Eve}

In the next section we show how to use use automata-based techniques to decide the existence of a (positional) winning strategy for Eve. These techniques rely on characterizing whether a positional strategy of Eve is surely, almost-surely and positively winning.

\subsubsection{Surely and almost-surely winning conditions}

We characterize (almost-)surely winning strategies.

\begin{definition}[$q$-branches]
Let $q \in Q$ and $\sigma$ a strategy.
An infinite branch of the binary tree is a $q$-branch in $\sigma$ 
if at least one $\sigma$-play which projects to this branch
has limsup $q$.
\end{definition}

\begin{lemma}\label{caracsure}
Assume the automaton has \what.
Let $\sigma$ be a positional strategy for Eve.
Then $\sigma$ is surely winning iff 
for every $q \in (Q\setminus \Qall)$
there is no $q$-branch in $\sigma$.
Moreover $\sigma$ is almost-surely winning iff for every $q \in (Q\setminus \Qone)$
the set of $q$-branches in
$\sigma$ has measure $0$.
\end{lemma}

\begin{proof}
We denote $\mu$ the uniform probability measure on $\{0,1\}^\omega$.
For every state $q$, $Y_q$ denotes the set of $q$-branches in $\sigma$. 

We show the first statement about sure winning.
For every $\sigma$-play $\pi$ there exists a strategy $\tau$ of Adam
such that $\pi$ is consistent both with $\sigma$ and $\tau$.
Thus there is $q\in (Q\setminus \Qall)$ such that $Y_q \neq \emptyset$
iff there is a strategy $\tau$ of Adam and 
a play consistent with $\sigma$ and $\tau$ with limsup in $Q\setminus\Qall$,
iff $\sigma$ is \emph{not} accepting.

We show that
the condition $\forall q \in Q \setminus \Qone, \mu(Y_q) = 0$ is sufficient for $\sigma$ to be almost-surely winning.
Let $\tau$ be a strategy of Adam
and $Y'$ the set of branches of plays consistent with $\sigma$ and $\tau$ which have limsup in $Q \setminus \Qone$.
Then $Y' \subseteq \bigcup_{q\in Q \setminus \Qone} Y_q$.
According to Lemma~\ref{probmu}, 
$\Proba^{\sigma,\tau}(\limsup \not \in \Qone)
=
\mu(Y')\leq \mu\left(\bigcup_{q\in Q \setminus \Qone} Y_q\right) = 0$.
Thus $\sigma$ is almost-surely winning.

We show that the condition $\mu(Y_q)>0$ for some $q\in Q \setminus \Qone$ is sufficient for $\sigma$ \emph{not} to be almost-surely winning.
For every infinite branch $b\in Y_q$ choose one $\sigma$-play $\pi_b$ with $\limsup \in Q\setminus \Qone$.
Since the automaton has \what, a suffix of $\pi_b$ is canonical, let $w_b$ be the first vertex of this suffix.
For every $\sigma$-reachable vertex $w$ denote $Z_{w} =\{ b \in Y_q \mid w_b = w \}$. 
Since $Y_q$ is the countable union of the sets $(Z_w)_{w\text{ $\sigma$-reachable}}$
there is at least one $\sigma$-reachable vertex $w$
such that $\mu(Z_{w})>0$.
Let $\pi_w$ be a finite $\sigma$-play to $w$.
Let $\tau_w$ a strategy for Adam which enforces $\pi_w$ with positive probability
and plays canonically in every continuation of $\pi_w$ whenever possible.
We show that $\Proba^{\sigma,\tau_w}(\limsup \not \in \Qone)>0$.
Let $X_w$ be the set of continuations of $\pi_w$ consistent with $\sigma$ and
$\tau_w$ whose branch belongs to $Z_w$.
An easy induction shows that every play $\pi'\in X_w$ with branch $b$ coincide with $\pi_b$ after $w_b$ ($\sigma$ is positional and 
$\tau_w$ plays only canonical moves).
Thus every play in $X_w$ has $\limsup \in Q \setminus \Qone$.
Then
$\Proba^{\sigma,\tau_w}(\limsup \not \in \Qone)
\geq
\Proba^{\sigma,\tau_w}(X_w)
= 
\mu(Z_w) > 0$,
according to Lemma~\ref{probmu}.
\end{proof}

Whether a branch is a $q$-branch can be checked by computing a system of $\sigma$-indexes.
Intuitively, all $\sigma$-reachable vertices receives a finite index, such that
along a $\sigma$-play the index does not change
except when Adam performs a non-canonical move or when two plays merge on the same vertex, in which case the smallest index is kept. After a non-canonical move of Adam, a new play may start in which case it receives a fresh index not used yet in the current neither in the parent node. For this less than $2|Q|$ indices are required.
The important properties of $\sigma$-indexes are:
\begin{lemma}[Characterization of $q$-branches]\label{defindex}
Every positional strategy $\sigma$ of Eve can be
associated with a function
\[
\index_\sigma:\{0,1\}^*\times Q \to \left\{0,1,\ldots ,2|Q|, \infty\right\}^Q
\]
 with the following properties.


First, $\index_\sigma$ can be computed on-the-fly along a branch.
For every node $n$ denote $\sigma_n$ the restriction of $\sigma$ on $\{n\}\times Q$.
 Then 
 $\index_\sigma(\epsilon)$ only depends on $\sigma_\epsilon$.
And for every node $n$ and $d\in\{0,1\}$,
 $\index_\sigma(nd)$ only depends on 
 $\index_\sigma(n)$ and $\sigma_{nd}$.

Second, a vertex  $(n,q)$ is reachable from the initial vertex by a $\sigma$-play
  iff $\index_\sigma(n)(q)$ is finite.

Third, let $b\in\{0,1\}^\omega$ be an
infinite branch of the binary tree,
visiting successively the nodes
$n_0,n_1,n_2, \ldots$.
Denote $R^\infty(b)$
the set of pairs 
$(k,q)\in \{0,\ldots, 2|Q|\}\times Q$ 
such that:
\begin{itemize}
\item
$k \in \index_\sigma(n_i)(Q)$
for every $i\in \NN$ except finitely many;
\item
and $k= \index_\sigma(n_i)(q)$
for infinitely many $i\in \NN$.
\end{itemize}

Then for every state $q$, the branch $b$ is a $q$-branch if
and only if
there exists $k\in \{0,1,\ldots ,2|Q|\}$
such that
$q=\max \{ r \in Q \mid (k,r) \in R^\infty(b)\}$.
\end{lemma}

\subsubsection{Checking the positively winning condition}

In order to check with a non-deterministic automaton
wheth\-er a positional strategy is positively winning,
we rely on the notion of \emph{positive witnesses}.
The point of positive witnesses is to turn the verification of up to $|Q|$ 
  positively-winning conditions
- depending on the decisions of Adam,
there may be up to $|Q|$ different $\sigma$-reachable vertices on a given node  -
into a single one. This single condition can then be  checked by a non-deterministic nonzero automaton
equipped with a single positively-winning condition.

\paragraph*{Everywhere thick subtrees.}
We need the notion of everywhere thick subtrees.
We measure sets of infinite branches with the uniform probability measure $\mu$ on $\{0,1\}^\omega$.

\begin{definition}[Subtree]
A set of nodes $T\subseteq \{0,1\}^*$ is a \emph{subtree}
if it contains a node $r$, 
called the root of $T$, 
such that every node $n\in T$ is a descendant of $r$,
$T$ contains all nodes on the path from $r$ to $n$.
\end{definition}

\begin{definition}[Everywhere thick sets of nodes]
For every set $T\subseteq \{0,1\}^*$ of nodes
denote $\vec{T}$ the set of branches in $\{0,1\}^\omega$
whose every prefix belongs to $T$.
Then $T$
is \emph{everywhere thick} if
starting from every node $n \in T$
there is nonzero probability to stay in $T$,
i.e. if $\mu\left(\vec{T} \cap n\{0,1\}^\omega\right)>0$.
\end{definition}

Everywhere thick subtrees are almost everywhere.

\begin{lemma}
\label{lem:thick}
Let $P\subseteq \{0,1\}^\omega$ be a measurable set of infinite branches.
Assume $\mu(P)>0$.
Then 
there exists an everywhere thick subtree $T$, with root $\epsilon$ such that $\vec{T} \subseteq P$.
 \end{lemma}

The proof relies on the inner-regularity of $\mu$,
so that $P$ can be assumed to be a closed set,
i.e. a subtree from which we can prune leaves whose subtree has probability $0$.

\paragraph*{Positive witnesses.}

Positive witnesses can be used to check whether a strategy is positively winning:

\begin{definition}[Positive plays and witnesses]
\label{defnonzero}
Let $\lab$ be a $\Sigma$-labelled binary tree
and $\sigma$ a positional strategy of Eve in the acceptance game of $\lab$.
Let $Z$ be the set of $\sigma$-reachable vertices whose state is in $\Qzero$.

A play is positive if all vertices it visits belong to $\{0,1\}^*\times \Qzero$.
A positive witness for $\sigma$ is a pair $(W,E)$
where:
\begin{align*}
 &W \subseteq Z \text{ are the \emph{active} vertices},\\
&E \subseteq \{0,1\}^* \times \{0,1\}\text{ is the set of \emph{positive edges},}
\enspace
\end{align*}
and they have the following properties.
\begin{itemize}
\item[a)]
From every vertex $z \in Z$ there is a positive and canonical finite $\sigma$-play starting in $z$ which reaches a vertex in $W$ or a transient vertex.
\item[b)]
Let $z=(n,q) \in W$.
Then $(n,0)\in E$ or $(n,1)\in E$, or both.
If $z \to z'$ is a local transition
then $z' \in W$ as well
whenever ($q\in Q_E$ and $z \to z'$ is consistent with $\sigma$)
or ($q\in Q_A$ and
$z\to z'$ is canonical).
If $z$ is controlled by Eve and $\sigma(z)$ is a split transition
 $q \to (q_0,q_1)$ then
$
((n,0)\in E \implies (n0,q_0) \in W)$
and
$((n,1)\in E \implies (n1,q_1) \in W)$.
 \item[c)]
 The set of nodes $\{ nd\in\{0,1\}^* \mid (n,d) \in E \}$
 is everywhere thick.
  \end{itemize}
\end{definition}

\begin{lemma}
[Characterization of positively winning strategies]
\label{lem:caracnonzero}
Assume the automaton has \what.
A positional strategy
 $\sigma$ for Eve is positively winning
iff there exists a positive witness for $\sigma$. 
\end{lemma}

\subsection{Deciding emptiness}

\newcommand{\sttree}{T}

A $\Sigma$-labelled
binary tree $\lab$ and a positional strategy $\sigma$ in the corresponding acceptance game
generate a tree
\[
\sttree_{\lab,\sigma} : \{0,1\}^*  \to (Q \cup Q\times Q)^{Q_E}\enspace.
\]
For every vertex $(n,q)$ controlled by Eve,
if $\sigma(n,q)$ is a local transition $q \to_{\lab(n)} q'$
then $\sttree_{\lab,\sigma}(n)(q)=q'$
and if
$\sigma(n,q)$ is a split transition $q \to_{\lab(n)} (q_0,q_1)$
then $\sttree_{\lab,\sigma}(n)(q)=(q_0,q_1)$.

\begin{definition}[Strategic tree]\label{def:st}
A tree $ \sttree: \{0,1\}^*  \to (Q \cup Q\times Q)^{Q_E}$ is \emph{strategic}
if there exists a tree  $\lab:\{0,1\}^* \to \Sigma$
and a positional strategy $\sigma$ for Eve such that
$\sttree=\sttree_{\lab,\sigma}$\enspace.
\end{definition}

We are interested in the strategic trees associated to winning strategies.
The rest of the section is dedicated to the proof of the following theorem.

\begin{theorem}\label{theo:recogstrat}
Fix an alternating nonzero automata with limited choice for Adam.
The language of strategic trees $\sttree_{\lab,\sigma}$ such that $\sigma$ wins the acceptance game of $\lab$ can be recognized by a non-deterministic nonzero automaton of size exponential in $|Q|$.
If $\Qall=Q$ in the alternating automaton, then the sure condition of the non-deterministic automaton is B\"uchi.
\end{theorem}

\begin{proof}
The characterizations of surely, almost-surely and positively winning strategies
given in lemmas~\ref{caracsure},~\ref{defindex} and~\ref{lem:caracnonzero}
can be merged as follows.

\begin{corollary}\label{carac}
Let $\sigma$ be a positional strategy $\sigma$ for Eve.
For every branch $b$ denote
\[
M(b)=\{ \max\{q \mid (k,q)\in R^\infty(b)\}
\mid  k\in 0\ldots 2|Q|\} \enspace.
\]
Then $\sigma$ is winning if and only if 
\begin{itemize}
\item
for every branch $b$, 
$M(b)\subseteq \Qall$;
\item
 and
for almost-every branch $b$, 
$M(b)\subseteq \Qone$;
\item
and there exists a positive witness for $\sigma$.
\end{itemize}
\end{corollary}

First of all, the non-deterministic automaton
$\ndautomata$ checks whether the input tree is a strategic tree,
for that it guesses on the fly the input tree $t :\{0,1\}^*\to \Sigma$
by guessing on node $n$ the value of $t(n)$
and checking that for every $q\in\QE$, $q\to_{t(n)} T(n)(q)$ is a transition of the automaton.

On top of that $\ndautomata$ 
checks the three conditions of Corollary~\ref{carac}.
For the first two conditions,
it
computes (asymptotically) along every  branch $b$
the value of
 $R^\infty(b)$ and thus of $M(b)$.
For that the automaton relies
on a Last Appearance Record memory (LAR)~\cite{Gurevich:1982:TAG:800070.802177}
whose essential properties are:
\begin{lemma}[LAR memory~\cite{Gurevich:1982:TAG:800070.802177}]
\label{LAR}
Let $C$ be a finite set of symbols.
There exists a deterministic automaton on $C$
called the \emph{LAR memory on $C$} with the following properties.
First, 
the set of states, denoted $Q$, has size $\leq |C|^{|C|+1}$
and is totally ordered.
Second, for every $u\in C^\omega$
denote $L^\infty(u)$ the set of letters seen infinitely often in $u$
and
$\lar(u)$ the largest state seen infinitely often during the computation on $u$.
Then $L^\infty(u)$ can be inferred from $\lar(u)$, precisely
there is a mapping $\phi : Q \to 2^C$ such that:
  $
  \forall u\in C^\omega, L^\infty(u) = \phi(\lar( u) )\enspace.
  $
  \end{lemma}

In order to compute $R^\infty(b)$ along a branch $b$,
the non-deterministic automaton $\ndautomata$
computes deterministically on the fly the $\sigma$-index
of the current node $n$, as defined in Lemma~\ref{defindex},
and 
 implements a LAR memory on the alphabet
\[
C=\{0,\ldots, 2|Q|\}\times( Q \cup \{\bot\})\enspace.
\]
 When visiting node $n$,
 $\ndautomata$ injects into the LAR memory all pairs 
$(\index_\sigma(q),q)$ such that $q\in Q$ and $\index_\sigma(q)\neq \infty$
plus all pairs $(k,\bot)$ such that $k \not \in \index_\sigma(n)(Q)$.
For every branch $b$, the set $R^\infty(b)$ is equal to all pairs $(k,q)$ seen infinitely often such that $(k,\bot)$ is seen only finitely often.
Thus, 
the LAR memory can be used to check the first two conditions of Corollary~\ref{carac}, more details are given at the end of the proof.

\smallskip

For now, we describe how the non-deterministic automaton $\ndautomata$ checks
 whether there exists a positive witness
$(W,E)$ (Definition~\ref{defnonzero}).
Denote by $Z$ the set of $\sigma$-reachable vertices whose state is in $\Qzero$.
On node $n$ the automaton guesses (resp. computes)  the vertices of $W$ (resp. $Z$) of the current node and guesses the elements of $E$ by storing three sets of states:
\begin{align*}
&W_n=\{ q \in Q \mid (n,q) \in W\}\\
&Z_n =
\{ q \in \Qzero \mid \index_\sigma(n,q) < \infty \}\\
&E_n=\{ b \in \{0,1\} \mid (n,b) \in E\}
\enspace.
\end{align*}
Then $\ndautomata$ checks all three conditions a), b) and c) in the definition of a positive witness as follows.

\smallskip

$\ndautomata$ checks condition a) in the definition of a positive witness
by guessing 
on the fly for every vertex in $Z$ a canonical positive $\sigma$-play to a vertex which is either transient or in $W$, in which case we say the canonical positive play \emph{terminates}.

For that $\ndautomata$
maintains 
an ordered list $P_n$ of states.
On the root node, $P_\epsilon$ is $Z_\epsilon \setminus W_\epsilon$.
When the automaton performs a transition,
it guesses for each state $q$ in $P_n$ and
direction $b_q$  
a successor $s_q$, such that $(nb_q,s_q)$ can be reached from $(q,n)$ by a positive canonical $\sigma$-play.
In direction $b$, every state $q$ for which $b_q \neq b$ is removed from the list,
while every state $q$ for which $b_q = b$
is replaced by the corresponding $s_q$.
 Then all states in $Z_{nb}$ are added at the end of the list.
In case of duplicates copies of the same state in the list,
only the first copy is kept. 
In case the head of the list is in $W_{nb}$ or is transient,
a B\"uchi condition is triggered
 and the head is moved at the back of the list.
Finally all entries of the list which are in $W_{nb}$ are removed.

This way, condition a) holds iff the B\"uchi condition is triggered infinitely often on every branch.
We discuss below how to integrate this B\"uchi condition
in the sure accepting condition of the automaton.

\smallskip

$\ndautomata$ checks condition b) in the definition of a positive witness
by entering an absorbing error state as soon as 
\begin{enumerate}[1)]
\item there is some local transition $(n,q)\to_{t(n)} (n,q')$ such that
 $q\in W_n$ and
($q\in Q_E$ and $z \to z'$ is consistent with $\sigma$)
or ($q\in Q_A$ and
$z\to z'$ is canonical); or
\item
there is some $q\in W_n$ controlled by Eve and $b\in E_n$ such that $\sigma(n,q)$ is a split transition $q \to_{t(n)}(q_0,q_1)$ but $q_b\not\in W_{nb}$. 
\end{enumerate}
The guessed sets $W_n$ are bound to satisfy condition 1) and condition 2) is checked by storing a subset of $Q$.

\smallskip

$\ndautomata$ checks condition c) in the definition of a positive witness
by triggering the positive acceptance condition
whenever it moves in direction $b$ on a node $n$ such that $b \in E_n$.

\smallskip

The sure and almost-sure acceptance condition are defined as follows.
The B\"uchi condition necessary for checking condition a) in the definition of a positive witness
is integrated in the LAR memory,
for that we add to the alphabet $C$ of the LAR memory a new symbol
$\top$ which is injected in the LAR memory  whenever the B\"uchi condition is triggered.
The order between states of $\ndautomata$ is induced by the order of the LAR memory. 

This way, according to Lemma~\ref{LAR},
the largest state seen infinitely often along a branch $b$ reveals
whether $\top$ was seen infinitely often,
and reveals  
the value of $R^\infty(b)$ (the set of pairs $(k,q)$ seen infinitely often
such that $(k,\bot)$ was seen finitely often) hence of $M(b)$ as well. The state is surely (resp. almost-surely) accepting 
iff $\top$ was seen infinitely often and $M(b)\subseteq \Qall$ (resp. $M(b)\subseteq \Qone$).
 In case $\Qall=Q$ in the alternating automaton then the sure condition boils down to the B\"uchi condition.

 According to Corollary~\ref{carac},
 and by construction of $\ndautomata$,
 the computation of $\ndautomata$ is accepting
 iff the input is a strategic tree whose corresponding
 strategy of Eve is winning.
\end{proof}

\begin{theorem}\label{theo:recogalt}
Emptiness of alternating nonzero automata with limited choice for Adam is decidable in \nexptime$\cap$co-\nexptime.
If $\Qall=Q$, emptiness can be decided in \exptime.
\end{theorem}
\begin{proof}
Emptiness of an alternating automaton reduces to the emptiness of a non-deterministic automaton of exponential size.
This non-deterministic automaton guesses on-the-fly a tree
$\{0,1\}^*  \to (Q \cup Q\times Q)^{Q_E}$ and checks it is a winning strategic tree,
using the automaton given by Theorem~\ref{theo:recogstrat}.
In case the alternating automaton is $\Qall$-trivial,
the sure condition of the non-deterministic automaton is B\"uchi (Theorem~\ref{theo:recogstrat}).
We conclude with Theorem~\ref{theo:complexitiyemptiness}.
\end{proof}


\section{Satisfiability of \ctlsa\ }\label{sec:pctls}

Our result on alternating nonzero automata 
can be applied to decide the satisfiability of the logic \ctlsa, a generalization of CTL* which integrates
both deterministic and probabilistic state quantifiers.

\paragraph*{Markov chains.}
The models of \ctlsa\ formulas are Markov chains.
A Markov chain
with alphabet $\Sigma$ is a tuple
$\mc=(\states,\pt,\lab)$ where
$\states$ is the (countable) set of \emph{states},
 $\pt : \states \to \distrib{(\states)}$ are the \emph{transition probabilities}
and $\lab:\states\to \Sigma$ is the \emph{labelling function}.

For every state $s\in S$,
there is a unique probability measure denoted $\Proba_{\mc,s}$ on $S^\omega$
such that
$\Proba_{\mc,s}(s S^\omega)=1$
and
for every sequence
$s_0\cdots s_n s_{n+1}\in S^*$,
$\Proba_{\mc,s}(s_0\cdots s_n s_{n+1} S^\omega)=\pt(s_n,s_{n+1})\cdot \Proba_{\mc,s}(s_0s_1\cdots s_n S^\omega)$.
When $\mc$ is clear from the context this probability measure is simply denoted $\proba_s$.
A \emph{path} in $\mc$ is a finite or infinite sequence of states $s_0s_1\cdots$ such that
$
\forall n \in \NN, p(s_n,s_{n+1})>0\enspace.
$. We denote $\pathes_{\mc}(s_0)$ the set of such paths.

A binary tree $\lab:\{0,1\}^*\to \Sigma$ is seen as a specific type of Markov chain,
where from every node $n\in \{0,1\}^*$
there is equal probability $\frac{1}{2}$ to perform transitions to $n0$ or $n1$.

\paragraph*{Syntax.}
For a fixed alphabet $\Sigma$,
there are two kinds of formula: state formula (typically denoted $\sf$) and path formula (denoted $\pf$),
generated by the following grammar:
\begin{align*}
\sf ::=& \top \mid \bot \mid a\in \Sigma \mid \sf \wedge\sf\mid \sf \vee\sf\mid \neg \sf \\
&\mid \PCTLE \pf \mid \PCTLA \pf \mid \Proba_{> 0}(\pf) \mid \Proba_{= 1}(\pf)\\
\pf ::=& \sf\mid\neg \pf \mid \pf \wedge \pf \mid \pf\vee \pf
\mid \next \pf \mid \pf \until \pf \mid \always\pf\enspace.
\end{align*}

\paragraph*{Semantics.}

Let $\mc=(\states,\lab,\pt)$ a Markov chain.
We define simultaneously and inductively the satisfaction
$
\mc,s \models \sf
$
 of a state formula $\sf$
by a state $s\in\states$ and the satisfaction
$
\mc,\aPath \models \pf
$
of a path formula $\pf$ by a path $\aPath\in\pathes_{\mc}$.
When $\mc$ is clear from the context,
we simply write $s \models\sf$ and $\aPath \models \pf$.

If a state formula is produced by one of the rules
$\top \mid \bot \mid p \mid \sf \wedge\sf\mid \sf \vee\sf\mid \neg \sf$,
its satisfaction is defined as usual.
If $\pf$ is a path formula and $\fo\in\{\exists\pf,\forall \pf, \Proba_{> 0}(\pf), \Proba_{=1}(\pf)\}$
then 
\begin{align*}
& s \models \PCTLE \pf &\text{ if } &\exists w \in \pathes_\mc(s),  w \models \pf \\
& s \models \PCTLA \pf &\text{ if } &\forall w \in \pathes_\mc(s), w \models \pf\\
& s \models \Proba_{\sim b}( \pf) &\text{ if } &\Proba_{\mc,s}(\aPath\in \pathes_\mc(s)\mid \aPath\models \pf)\sim b\enspace.
\end{align*}

The satisfaction of a path formula $\pf$ 
by an infinite path $\aPath=s_0s_1\dots\in
\pathes_{\mc}(s_0)$
is defined as follows.
If $\pf$ is produced by one of the rules
$\neg \pf \mid \pf \wedge \pf \mid \pf\vee \pf$
then its satisfaction is defined as usual.
If $\pf$ is a state formula (rule $\pf:=\sf$)
then 
$
\aPath \models \sf$ if $s_0 \models \sf\enspace.
$
Otherwise, 
$\pf \in
\{
\next \pf',\always \pf', \pf_1\until \pf_2
\}$
where
$\pf', \pf_1$ and $\pf_2$ are path formulas.
For every integer $k$,
we denote $w[k]$ the path
$s_ks_{k+1}\dots\in \pathes_{\mc}(s_k)$.
Then:
\begin{align*}
&\aPath \models \next \pf' &\text{ f } &w[1] \models \pf'\\
&\aPath \models \always \pf' &\text{if } &\forall i\in\nats, w[i]\models \pf'\\
&\aPath \models \pf_1\until \pf_2 &\text{if }&\exists n \in \nats, (\forall 0\leq i < n, w[i] \models \pf_1 \land w[n] \models \pf_2).
\end{align*}

The Markov chain given in Figure~\ref{fig:mc} satisfies the formula $(\PCTLA( \always \PCTLE (\top \until a) ))\wedge 
(\PCTLPp( \always \neg a))$.

\begin{figure}
	\begin{tikzpicture}[node distance=1.3cm]
	\node (1) {b};
	\node (11) [above right of=1] {a};
	\node (12) [below right of=1] {c};
	\node (2) [below right of=11]{b};
	\node (21) [above right of=2] {a};
	\node (22) [below right of=2] {c};
	
	\node (d) [below right of=21]{\dots};
	\node (311) [right of=d]{b};
	\node (111) [above right of=311] {a};
	\node (121) [below right of=311] {c};
	
	\node (d) [below right of=111]{\dots};

	\path[->,sloped,above]
	(1) edge node {$\frac{1}{2^2}$}(11)
	(1) edge node {$1-\frac{1}{2^2}$}(12)
	(11) edge (2)
	(12) edge (2)
	(2) edge node {$\frac{1}{2^3}$}(21)
	(2) edge node {$1-\frac{1}{2^3}$}(22)
	
	(311) edge node {$\frac{1}{2^n}$}(111)
	(311) edge node {$1-\frac{1}{2^n}$}(121)
	;
	\end{tikzpicture}
	
	\caption{A model of $(\PCTLA( \always \PCTLE (\top \until a) ))\wedge 
		(\PCTLPp( \always \neg a))$\label{fig:mc}}
\end{figure}

\paragraph*{Variants and fragments.}

A formula of \ctlsa\ belongs to the fragment \ctl\
if in each of its state subformula $\sf$ of type
$\PCTLE \pf \mid \PCTLA \pf \mid \Proba_{> 0}(\pf) \mid \Proba_{= 1}(\pf)$ the path formula $\pf$ has type $\next \sf' \mid \sf'\until \sf'' \mid \always\sf'$ where $\sf'$ and $\sf''$ are state subformulas.

In the variant \ECTL, every path formula $\pf$ is described as the composition of a deterministic B\"uchi automata on some alphabet $\{0,1\}^k$ with $k$ state subformulas. A path satisfies $\pf$ if the B\"uchi automaton accepts the sequence of letters obtained by evaluating the $k$ state subformulas on every state along the path.
This variant augments both the expressivity and the conciseness of the logic at the cost of a less intuitive syntax. For more details see~\cite{brazdil2008controller}. 

We are also interested in the fragments where the operators $\PCTLE$ and $\PCTLA$ are not used,
 i.e. the qualitative fragments of the logics \pctls, p\ECTL\ and p\ctl.

\paragraph*{Satisfiability problem.}
 A Markov chain $\mc$ \emph{satisfies} a formula $\pCTLFormula$
 at state $s$,
 or equivalently $(\mc,s)$ \emph{is a model} of $\pCTLFormula$,
 if $\mc,s \models\pCTLFormula$.
We are interested in the problem:

\smallskip

\noindent\SP: given a formula, does it have a model?

\smallskip

The satisfiability of {\sc tmso+zero} is known to be decidable~\cite{DBLP:conf/icalp/Bojanczyk16, DBLP:journals/corr/BojanczykGK17}. 
Since \ctlsa\ is a fragment of {\sc tmso+zero},
\SP\ is decidable
with non-elementary complexity.
A reduction  to the emptiness of alternating nonzero automata gives better complexity:

\begin{theorem} \label{theo:pctls}
For
\ctlsa\ the satisfiability problem is in 3-\nexptime\ $\cap$ co-3-\nexptime.
The following table summarizes complexities of the satisfiability problem for various fragments and variants of  \ctlsa:

\smallskip

\begin{tabular}{|c|c|c|}
\hline
& \allop & \pop\\
\hline
\ctls 
& \makecell{3-\nexptime\\$\cap$ co-3-\nexptime} 
& \makecell{3-\exptime~\cite{brazdil2008controller}\\ (qualitative \pctls)}\\
\hline
\ECTL 
& \makecell{2-\nexptime\\$\cap$ co-2-\nexptime} 
& \makecell{2-\exptime~\cite{brazdil2008controller}\\(qualitative p\ECTL)}\\
\hline
\ctl 
& \makecell{\nexptime\\$\cap$ co-\nexptime} 
& \makecell{\exptime~\cite{Brazdil2008}\\
(qualitative p\ctl)}\\
\hline
\end{tabular}
\end{theorem}

According to~\cite{Brazdil2008,brazdil2008controller}, the complexities
for \ECTL\pop\ and \ctl\pop\ are optimal.

The first step in the proof of Theorem~\ref{theo:pctls}
is a linear-time reduction from \SP\ to:

\smallskip
\noindent \BSP:
given a formula, does it have a model among binary trees?

\begin{theorem}\label{theo:reduc}
Any formula $\fo$ of \ctlsa\ on alphabet $\Sigma$ 
can be effectively transformed into a formula $\fo'$
of linear size on alphabet $\Sigma\cup \{\blank\}$ such that
$\fo$ is \SP\ iff $\fo'$ is \BSP.
As a consequence, \SP\ linearly reduces to \BSP.
 This transformation stabilizes
the fragment \ctls$[\PCTLPp,\PCTLPo]$.
\end{theorem}


The second step is a standard translation from logic to alternating automata~\cite{ltl}.

\begin{lemma}\label{lem:pctltobc}
For every  formula $\pCTLFormula$ of \ctlsa\ (resp. \ECTLa),
there is an alternating  automaton $\mathcal{A}$ with \what\
whose language is the set of binary trees
satisfying the formula at the root.
The automaton 
is effectively computable, of
size $O(2^{2^{|\pCTLFormula|}})$ (resp.  $\mathcal{O}({2^{|\pCTLFormula|}})$).
If $\pCTLFormula$ is a \ctl\ formula,
the size of $\mathcal{A}$ is $\mathcal{O}({{|\pCTLFormula|}})$.
In case the formula does not use the $\PCTLE$ and $\PCTLA$ operators, the $\Qall$ condition is trivial i.e. $\Qall=\Q$.
\end{lemma}

\begin{proof}[Proof of Theorem~\ref{theo:pctls}]
All the complexity results are obtained by reduction of \SP\ to the emptiness problem for an alternating nonzero automaton with limited choice for Adam, which is decidable in \nexptime$\cap$co-\nexptime\
(Theorem~\ref{theo:recogalt}). The size of the automaton varies from doubly-exponential to linear size depending whether the formula is in  \ctls, \ECTL\ or \ctl\ (Lemma~\ref{lem:pctltobc}).
In case the formula does not use the deterministic operators $\PCTLE$ and $\PCTLA$ (i.e. for qualitative  \pctls, p\ECTL\ and p\ctl) the $\Qall$ condition of the alternating automaton is trivial thus its  emptiness is decidable in \exptime\
(Theorem~\ref{theo:recogalt}).
\end{proof}

\section*{Conclusion}
We have introduced the class of \emph{alternating nonzero} automata,
proved decidability of the emptiness problem for the subclass
of automata with \what\
and obtained as a corollary algorithms for the satisfiability of a  
temporal 
logic extending both CTL* and the qualitative fragment of pCTL*.

A natural direction for future work is to find more general classes of alternating nonzero automata with a decidable emptiness problem, which requires some more insight on the properties of the acceptance games.


\newpage

\appendix

\section*{Appendix}


\section{Closure properties}
\begin{proof}[Proof of Lemma~\ref{lem:closure}]
Take the disjoint union of the automata,
plus an initial state which is not in $\Qzero$ and two transitions leading to the initial states of
the original automata, controlled by Adam
for intersection and Eve for union.
\end{proof}

\section{Emptiness problem in the non-deterministic case: a proof of Theorem~\ref{theo:complexitiyemptiness}}

In this section we prove Theorem~\ref{theo:complexitiyemptiness}.

For the sake of completeness,
we provide an algoritm to decide emptiness of nonzero automata in \npconp\ 
and show that complexity drops to {\sc ptime} when $\Qall$ is a B\"uchi condition.
In the sequel we fix a non-deterministic automaton 
\[
\automata=(\Q,q_0,\QE,\Atrans,\Qall, \Qone,\Qzero)\enspace.
\]

Recall that in a non-deterministic automaton is the special case of alternating nonzero automata
where all transitions are controlled by Eve and are split transitions.

\paragraph*{Single-letter alphabets}

In the case where the input alphabet has a single letter,
there is a single possible input tree,
thus a single possible acceptance game,
which makes things easier.

There is a linear-time reduction of the emptiness problem for non-deterministic nonzero automata
to the special case of single-letter alphabet:
if one has an algorithm for the latter problem,
then the former problem can be solved by having the non-deterministic automaton
 guesses the letters of the input tree on the fly, and perform its computation as usual.

In the sequel we assume the alphabet contains a single letter.

\paragraph*{Accepting runs}

In a non-deterministic automaton,
only Eve takes decisions.
Once the strategy of Eve is fixed,
there is for every node $n$ a single vertex $(n,q)$ which is $\sigma$-reachable,
and the strategy of Eve can be represented as a mapping $S: \{0,1\}^* \to \Q$ such that $S(\epsilon)=q_0$
and for every node $n$,
$(S(n),S(n0),S(n1))$ is a transition of the automaton.
Such a mapping is called a \emph{run} of the automaton.

Notions of sure, almost-sure and positive acceptance extend naturally to runs
(see~\cite{DBLP:journals/corr/BojanczykGK17} for full details).

\paragraph*{Almost-sure policies}

\begin{definition}[Policy]
Let $W\subseteq \Q$. 
A policy with domain $W$
is a mapping $\sigma:W\to W^2$ such that 
for every $q\in W$,  $q \to \sigma(q)$
is a (split) transition of the automaton.
\end{definition}

A policy $\sigma$ induces a Markov chain 
$\mathcal{M}_\sigma$
with states $W$, whose transition probabilities are defined as follows.
From every state $q\in W$ with $\sigma(q)=(q_0,q_1)$,
if $q_0\neq q_1$ there is probability $\frac{1}{2}$ to go to 
either states $q_0$ or $q_1$
and if $q_0=q_1$ there is probability $1$
to go to the state $q_0=q_1$.

\begin{definition}[Almost-sure policies]
A policy $\sigma$ with domain $W$
is \emph{almost-sure} if:
\begin{itemize}
\item
the maximal state of every ergodic class of $\mathcal{M}_\sigma$
is in $\Qone$,
\item
from every state in $\Qzero\cap W$ there is a path  in $\mathcal{M}_\sigma$ staying 
in $\Qzero\cap W$ and reaching an ergodic class included in $\Qzero$.
\end{itemize}
\end{definition}

We will use twice this auxiliary lemma.

\begin{lemma}\label{largest}
Let $\sigma_1$ and $\sigma_2$ be two almost-sure policies with domains $W_1$ and $W_2$.
Then there exists an almost-sure policy with domain $W_1 \cup W_2$.
\end{lemma}
\begin{proof}
We define 
$
\sigma_3(w) =
\begin{cases}
\sigma_1(w) \text{ if } w \in W_1\\
\sigma_2(w) \text{ otherwise}
\end{cases}
\enspace.
$
%
Then $\sigma_3$ is a almost-sure policy.
For that, remark that once a play in the Markov chain $\mathcal{M}_{\sigma_3}$
enters $W_1$ it stays in $W_1$.
Thus every ergodic class of $\mathcal{M}_{\sigma_3}$ is an ergodic class of either $\sigma_1$ or $\sigma_2$.
\end{proof}

The notion of almost-sure policy is similar to the notion of acceptance witnesses
in~\cite[Definition 11]{DBLP:journals/corr/BojanczykGK17},
 where is established the following result:
\begin{theorem}\cite[Theorem 10]{DBLP:journals/corr/BojanczykGK17}
\label{ptime}
Assume the non-deterministic automaton is $\Qall$-trivial i.e. $\Qall=Q$.
Then it is non-empty iff there exists an almost-sure policy whose domain contains
the initial state of the automaton. This is decidable in \ptime.
\end{theorem}

\paragraph*{Perfect plays}

It is an exercice to design a sequence of integers $0 = n_0 \leq n_1 \leq n_2 \leq \ldots$
which is (very slowly) converging  to $\infty$
and such that, starting from any vertex of any ergodic component of $\mathcal{M}_\sigma$,
there is at least probability $\frac{1}{2}$ that for every $k$, at step $k$ every state of the ergodic component
has already been visited at least $n_k$ times. This sequence can be constructed either explicitely,
by elementary calculations, or using the inner-regularity of the probability measure (cf.~\cite[Theorem 17.10]{kechris}).

Let $\mathcal{M}_\sigma$ be a Markov chain  induced by an almost-sure policy $\sigma$.
A \emph{perfect play} of $\sigma$ is a path in $\mathcal{M}_\sigma$ with the following properties:
\begin{itemize}
\item[a)]
whenever the path enters $\Qzero$, its stays in $\Qzero$ afterwards,
\item[b)]
each step of the path strictly reduces the distance to the ergodic classes , until such a class is reached,
\item[c)]
$k$ steps after the entry inside an ergodic class,
every state of the ergodic class has been visited at least $n_k$ times.
\end{itemize}

\begin{lemma}[Perfect plays happen with $>0$ probability]\label{lem:perfect}
For every almost-sure policy $\sigma$ and every state $q$ of its domain, the set of perfect plays starting from $q$ has  probability
$\geq \frac{1}{2^{|Q|+1}}$.
\end{lemma}
\begin{proof}
Properties a) and b) hold with probability at least $\frac{1}{2^{|Q|}}$
because by definition of almost-sure policies,
there is at least one direction which reduces the distance to the ergodic classes,
and moreover stays in $\Qzero$ once it enters $\Qzero$.
Property c) holds with probability at least $\frac{1}{2}$ by choice of $n_1,n_2,\ldots$.
\end{proof}

\paragraph*{The $\Qall$-B\"uchi case}

\begin{lemma}\label{caracbuc}
Assume $\Qall$ is a B\"uchi condition.
Then its language is non-empty
iff
there exists $W\subseteq Q$ and two policies $\sigma_1,\sigma_2$ with domain $W$
such that:
\begin{itemize}
\item $W$ contains the initial state of the automaton.
\item $\sigma_1$ is almost-sure.
\item $\sigma_2$ guarantees infinitely many visits to $\Qall$
in the sense where every infinite play in $\mathcal{M}_{\sigma_2}$ visits $\Qall$ infinitely often.
\end{itemize}
\end{lemma}
\begin{proof}
Similar proof techniques have been independently used in the framework of \emph{beyond worst-case synthesis}~\cite{DBLP:journals/corr/BerthonRR17}.

We first show that the conditions are sufficient.
Eve can win by combining $\sigma_1$ and $\sigma_2$ into a strategy $\sigma$ defined
as follows.
Eve starts with playing $\sigma_1$ and keeps playing $\sigma_1$
as long as the play is perfect. 
In case the play is not perfect anymore,
 Eve switches to $\sigma_2$ until $\Qall\cap W$ is reached.
Once $\Qall \cap W$ has been reached, Eve 
switches again to $\sigma_1$ and keeps playing $\sigma_1$
as long as the suffix of the play since the last switch is perfect.
Since there is probability $\geq \frac{1}{2^{|Q|+1}}$ that a play consistent with $\sigma_1$
 is perfect (Lemma~\ref{lem:perfect}),
almost-surely Eve switches only finitely many times to $\sigma_2$,
thus almost-surely a suffix of the play is perfect and consistent with $\sigma_1$.

The strategy $\sigma$ is almost-surely winning  because almost every play
consistent with $\sigma$ has a perfect suffix.
Such a suffix enters an ergodic component of $\mathcal{M}_{\sigma_1}$
and visits all its states infinitely often.
Since $\sigma_1$ is almost-surely winning,
the maximal state of every of its ergodic components is in $\Qone$,
thus almost-every play consistent with $\sigma$ has limsup in $\Qone$.

Moreover  $\sigma$ is positively winning:
the automaton is normalized thus every finite play $\pi$ reaching $\Qzero$
has at least one infinite continuation $\pi'$ in $\Qzero$.
In $\pi'$ Eve eventually switches to $\sigma_1$ and 
with positive probability from 
this moment on the play is perfect thus eventually stays in an ergodic component of $\mathcal{M}_{\sigma_1}$
intersecting $\Qzero$.
Such a component is actually included in $\Qzero$ because $\sigma_1$ is positively winning.


And $\sigma$ is surely-winning.
 There are two types of plays.
In plays where Eve switches infinitely often from one strategy to the other then infinitely many $\Qall$ states are visited,
because switches from $\sigma_2$ to $\sigma_1$ occur precisely under this condition.
And plays where Eve switches finitely often have a perfect suffix,
and as already seen these plays have limsup in $\Qone$,
which is included in $\Qall$ since the automaton is normalized.

\medskip

To show that the conditions are necessary,
start from a surely, almost-surely and positively accepting run $S: \{0,1\}^* \to \Q$.
Let $W$ be the image of $S$.

Remark that for every node $n$, the run $S_n : m \to S(nm)$ is
also almost-surely and positively accepting.
Then according to Theorem~\ref{ptime},
for every state $q\in W$,
Eve has an almost-sure policy $\sigma_q$ whose domain contains $q$.
According to Lemma~\ref{largest},
Eve has an almost-sure policy $\sigma_1$ whose domain contains $W$.

Now we construct $\sigma_2$.
Since $S$ is surely accepting, then every branch of $S$ visits $\Qall$ infinitely often.
As a consequence, Eve wins the B\"uchi game to $\Qall$ where Eve chooses the transitions and Adam chooses the direction. The policy $\sigma_2$ is a positional winning strategy for Eve in this B\"uchi game.
\end{proof}

\paragraph*{Proof of Theorem~\ref{theo:complexitiyemptiness}}

The \npconp\ upperbound is established by~\cite[Theorem 3]{DBLP:journals/corr/BojanczykGK17}.

\medskip

%
%
%
%

%

For the B\"uchi case, we show that the caracterization given in  Lemma~\ref{caracbuc}
can be decided in polynomial time, thanks to a fixpoint algorithm.

We can compute in polynomial time the largest domain $X(W) \subseteq W$ of an almost-sure policy $\sigma_1$,
which exists according to Lemma~\ref{largest}.
For every $w\in W$ and  $W \subseteq Q$, we denote $\mathcal{A}_{w,W}$
the $\Qall$-trivial automaton with initial state $w$ and restricted to states in $W$ and transitions in $W\times W^2$.
For that for every $w\in W$ we check whether the language of $\mathcal{A}_{w,W}$ is empty
or not, which can be done in \ptime\ according to~\cite[Theorem 10]{DBLP:journals/corr/BojanczykGK17}.
Then $X(w)$ is exactly the union of all $w$ for which $\mathcal{A}_{w,W}$ has a non-empty language.

We can also compute in polynomial time the largest domain $Y(W) \subseteq W$
of a policy $\sigma_2$ which guarantees infinitely many visits
to $\Qall$.
This simply amounts to computing the winning vertices
of the B\"uchi game played on $W$ where from state $w\in W$
Eve chooses any transition $w \to (w_0,w_1)$ with both $w_0\in W$ and $w_1\in W$,
loses if there is no such transition,
and Adam selects either $w_0$ or $w_1$.
The play is won by Eve iff $\Qall$ is visited infinitely often.

Now consider the largest fixpoint $W_\infty$ of the monotonic operator on $2^Q$ defined by $W \to Y(X(W))$\enspace.
We claim that $W_\infty$ contains the initial state of the automaton
if and only if  the characterization given in Lemma~\ref{caracbuc} holds.
Since $W_\infty$ is a fixpoint, then $W_\infty= X(W_\infty)=Y(W_\infty)$
thus there are policies $\sigma_1$ and $\sigma_2$ with domain $W_\infty$ which satisfy the characterization.
Conversely, if there is $W$ and $\sigma_1$ and $\sigma_2$ which satisfy the characterization
then $W = X(W) = Y(W)$ thus by monotonicity , $W \subseteq W_\infty$.
\qed

\section{Proof of Lemma~\ref{NLL}}

Let $\pi$ be a path  which arrives for the first time on a node $n$ in a vertex $(n,q)$, such that $q$ belongs to some connected component $C_q$. Let $P$ be the player controlling $q$.

The exit-profile of $\pi$ in $\sigma$ is the set of pairs $(q',t)$ such that there is an extension of $\pi$ consistent with $\sigma$ which leaves node $n$ in which $t$ is the last transition played before leaving the node and
$q'$ is the largest state seen on the path from $(n,q)$ up to now.
The loop-profile is the set of states $q'$ such that there is an extension of $\pi$ consistent with $\sigma$ which never leaves the current node $n$ and has limsup $q'$.
The set of possible exit-profiles and loop-profiles actually implementable by a strategy is easy to precompute,
since only one bit of memory is needed to implement any implementable profile.

Then in $\automata'$, $P$ announces an exit-profile and a loop-profile that she effectively can implement 
Her opponent $P'$ picks up an element of these profiles.
If this element is a cycle profile $q'$ then the new state is the sink state $(q',*)$ and the only transition available from there is a split transition staying in $(q',*)$ on both directions.
If this element is an exit profile $(t,q')$ then the transition $t$ is performed.
The accepting sets are adapted in $\automata'$ in order to have the natural correspondence between plays in $\automata$ and plays in $\automata'$ which preserves sure, almost-sure and positive winning conditions.

\section{Positional determinacy for Eve}

\begin{proof}[Proof of Lemma~\ref{lem:det}]

Fix an input tree $\lab$ and an automaton $\automata$.

We show that if Eve has a winning strategy $\sigma$ then she has a positional winning strategy. 
A finite $\sigma$-play $\pi$ has a \emph{canonical $\sigma$-extension} to some vertex $w$ if $\pi$ has a continuation $\pi'$ consistent with $\sigma$, whose last vertex is $w$ and in which
Adam only play canonical moves after the prefix $\pi$. In other words, if after $\pi$ Eve continues to play $\sigma$ and Adam is bound to play only canonical moves then there is positive probability that the play reaches $w$.

We equip the set of finite plays with any total order $\preceq$ such that the shorter is a play the smaller it is.

Let $P$ be a player and $\sigma$ a winning strategy for $P$.

For every vertex $w$, denote
$f(w)$ the smaller play for the order $\preceq$ among the plays 
starting in the initial vertex which have a canonical $\sigma$-extension  to $w$.
According to the (NLL) hypothesis,
the game graph is acyclic, thus there are finitely many plays 
having a canonical $\sigma$-extension  to $w$
and $f(w)$ is well-defined.

Denote $l(w)=\last(f(w))$.
Let $g(w)$ be a finite play such that $f(w)g(w)$ is a canonical $\sigma$-extension of $f(w)$ to $w$.
Remark that $g(w)$ is uniquely defined:
this is the unique continuation of $f(w)$ in which Eve plays $\sigma$, Adam plays only canonical moves and 
the play follows the branch connecting the node of $l(w)$ to the node of $w$.
For every vertex $w$ we define
\[
h(w) = f(w)g(w)\enspace.
\]

The winning positional strategy $\sigma'$ of Eve is defined by:
\[
\sigma'(w)= \sigma(h(w))\enspace.
\]
This strategy $\sigma'$ is well-defined because its definition inductively guarantees that a play consistent with $\sigma'$ visits only $\sigma$-reachable vertices.

\medskip

Let $\pi=w_0,w_1,\ldots = (n_0,q_0)(n_1,q_1),\ldots$ be an infinite play consistent with $\sigma'$.

Remark first that,
\begin{itemize}
\item[$(\dagger)$]
if all moves of Adam between some dates $i \leq j$ are canonical,
then $f(w_t) _{i \leq t \leq j}$ is non-increasing
for the $\preceq$-order.
\end{itemize}
The reason is that 
for every $i \leq t < j$, the concatenation
of  $g(w_t)$ and $w_{t+1}$
is a canonical $\sigma$-extension of $f(w_t)$ to $w_{t+1}$.
Hence by minimality of $f(w_{t+1})$, we get
 $f(w_{t+1})\preceq f(w_t)$.

 The sequence $(q_i)_{i\in\NN}$ is a path in the transition graph $G_\to$ thus it ends up in a connected component $C$ of $G_\to$. Let $t_\pi$ be minimum such that $q_{t_\pi}\in C$. Then after date $t_\pi$ all moves of Adam in $\pi$ are canonical.
According to  $(\dagger)$,
the sequence $(f(w_i))_{i\geq t_\pi}$ 
is non-increasing. And the shorter is a play the smaller it is
thus according to (NLL) this sequence takes finitely many values.
Thus it becomes constant 
after some index, denote it $k_\pi$
and denote $P(\pi)$ the prefix of $\pi$ of length 
$k_\pi$ 
and $S(\pi)$ the suffix of $\pi$ such that $\pi=P(\pi)S(\pi)$.
By definition of $\sigma'$ and since $k_\pi\geq t_\pi$,
\medskip
\begin{itemize}
\item[$(\dagger\dagger)$]
$S(\pi)$ is canonical and 
$h(w_{k_\pi})S(\pi)$ is a $\sigma$-play.
\end{itemize}
\medskip

We show that $\sigma'$ is winning.
%
%
%
We fix for the rest of the proof a strategy $\tau$ for Adam and show that $\sigma'$ is almost-surely and positively
winning against $\tau$.

For every finite play $\pi_1$
denote 
\[
W(\pi_1) = \{ S(\pi) \mid \pi \text{ a play consistent with $\sigma'$ and $\tau$ such that }  P(\pi)=\pi_1\}\enspace.
\]
Then,
according to $(\dagger\dagger)$,
\medskip
\begin{itemize}
\item[$(\dagger\dagger\dagger)$]
every play $\pi \in W(\pi_1)$  is canonical
and $h(\last(\pi_1))\pi$
is a $\sigma$-play.
\end{itemize}

As a consequence, $\sigma'$ is surely-winning: let $\pi$ be an infinite $\sigma'$-play then $S(\pi)\in W(P(\pi))$
thus according to $(\dagger\dagger\dagger)$ the play $S(\pi)$ is the suffix of a $\sigma$-play and since $\sigma$ is surely-winning the limsup of $S(\pi)$ is in $\Qall$ hence the limsup of $\pi$ as well since $\pi=P(\pi)S(\pi)$.


Fix some finite play $\pi_1$ consistent with $\sigma$ and $\tau$ 
and focus on $W(\pi_1)$.
Let $\pi_1'=h(\last(\pi_1))$.
Let $\tau'$ a strategy for Adam such that $\pi_1'$
is consistent with $\tau'$ and after $\pi_1'$ happens,
$\tau'$ performs canonical moves whenever possible.
We show that
for every measurable set of plays $E$,
\be\label{theq}
\proba^{\sigma',\tau}\left( \pi_1W(\pi_1)\cap \pi_1 E\right)
=
\proba^{\sigma,\tau'}\left(  \pi_1' W(\pi_1)
\cap \pi_1' E \right).
\ee
According to $(\dagger\dagger\dagger)$,
for every play $\pi_2\in W(\pi_1)$,
$\pi_1\pi_2$ is consistent with $\sigma'$ and $\tau$
and $\pi'_1\pi_2$ is consistent with $\sigma$ and $\tau'$.
Moreover, since $\pi_1$ and $\pi'_1$ have the same last vertex $\last(\pi_1)$ then
the set of $(\sigma',\tau)$-plays in
$\pi_1W(\pi_1)\cap \pi_1 E$
and the set of
$(\sigma,\tau')$-plays in 
$\pi_1' W(\pi_1)
\cap \pi_1' E$ project to the same set of branches.
Thus according to Lemma~\ref{probmu} the two probabilities coincide.


\medskip

We can use~\eqref{theq} to show that $\sigma'$ is almost-surely winning against $\tau$.
When $E$ is the event $\limsup \not \in \Qone$,
since $\sigma$ is almost-surely winning,
we get from~\eqref{theq} that $\proba^{\sigma',\tau}\left( \pi_1W(\pi_1)\cap \{\limsup \not \in \Qone\}\right)=0$.
There are countably many sets $\pi_1W(\pi_1)$,
these sets are measurable and their union contains all infinite $\sigma'$-plays. 
Thus $\proba^{\sigma',\tau}\left(  \{\limsup \not \in \Qone\}\right)=0$ hence $\sigma'$ is almost-surely winning.

\medskip 

Now we prove that $\sigma'$ is positively winning.
Let $\pi_2$ be a play consistent with $\sigma'$ and $\tau$
whose last state is positive.
Let $\pi_3$ be a positive continuation
of $\pi_2$ whose last state is minimal for the topological order on the transition graph of the automaton. 
Then in every positive continuation of $\pi_3$, Adam has no choice but playing canonical moves. 
According to ($\dagger$), the value of $f$ decreases along these continuations, let $\pi_4$ be a positive continuation of $\pi_3$ which minimizes $f(\last(\pi_4))$ among 
all positive continuations of $\pi_3$.
Then for every positive play $\pi$,
$\pi_4\pi$ is consistent with $\sigma'$ and $\tau$ 
iff $h(\last(\pi_4))\pi$
is consistent with $\sigma$ and $\tau'$,
where $\tau'$ is a strategy playing canonical moves whenever possible.
As a consequence, using Lemma~\ref{probmu},
\be\label{theq2}
\proba^{\sigma',\tau}\left( 
\{\text{ positive extensions of $\pi_4$}\}
\right)
=
\proba^{\sigma,\tau'}\left(
\{\text{ positive extensions of $h(\last(\pi_4))$}\}
\right)\enspace.
\ee
Since $h(\last(\pi_4))$ is a $\sigma$-play whose last vertex
is $\last(\pi_4)$ the right handside is positive thus the left handside as well.
 Thus $\sigma'$ is positively winning.
\end{proof}

\section{Characterization of $q$-branches (Proof of Lemma~\ref{defindex})}
\begin{proof}
For every vertex $w=(n,q)$, $\index_\sigma(n)(q)$
is defined by induction on the game graph,
which is acyclic thanks to the {\bf (NLL)}
assumption.
First in two simple cases:
\[
\index_\sigma(n)(q)=
\begin{cases}
0 & \text{ if $(n,q)$ is the initial vertex,}\\
\infty & \text{ if $(n,q)$ is not reachable by a $\sigma$-play.}
\end{cases}
\]
Then denote $A$ the set of immediate predecessors of $(n,q)$ by a canonical $\sigma$-play i.e. all the vertices $(n',q')$ such that:
\begin{itemize}
\item
$q'\in \QA$ and $(n',q')\to(n,q)$ is the canonical  local transition; or
\item
$q'\in \QE$ and $\sigma(n',q')$ is the local transition 
$(n',q')\to(n,q)$;
or
\item
$q'\in \QE$ is controlled by Eve and $\sigma(n',q')$ is a split transition whose
 $(n,q)$ is one of the two targets\enspace.
\end{itemize}
If $A\neq \emptyset$ then $
\index_\sigma(n)(q) = \min_{v \in A}  \index_\sigma(v)\enspace$.
Otherwise $\index_\sigma(n)(q)$
is the smallest index not attributed yet to any vertex whose node is either the node  $n$ or its father.

We prove the third property.
Assume first that the branch $(n_i)_{i\in\NN}$ is a $q$-branch.
Let $\pi=(n_0,q_0)(n_1,q_1)\cdots$ be a $\sigma$-play which projects to this branch and has limsup $q$.
For every $i\in\NN$ denote $k_i=\index_\sigma(n_i,q_i)$. Then $k_i\neq \infty$ because every vertex $(n_i,q_i)$ is $\sigma$-reachable.
Adam performs finitely many non-canonical moves in $\pi$
and after the last one of them,
 the sequence $(k_i)_{i\in\NN}$ is decreasing,
hence converges to some limit $k_\infty$.
From the moment this limit is reached, $\index_\sigma(n_i)(q_i)=k_\infty$
thus $k_\infty \in \index_\sigma(n_i)(Q)$
and $\max \{ r \in Q \mid (k_\infty,r) \in R^\infty(b)\}=
\limsup q_i = q$.

Now let $(k,q)\in R^\infty(b)$
such that 
$q=\max \{ r \in Q \mid (k,r) \in R^\infty(b)\}$.
By definition of $R^\infty$,
there exists $i_0$ large enough such that $\forall i \geq i_0,
k \in \index_\sigma(n_i)(Q)$
and by definition of $\index_\sigma$,
one can build by induction 
a $\sigma$-play $\pi$
starting on node $n_i$ 
in which Adam plays only canonical moves
and which visits exactly the vertices
$\{ (n_i,r) \mid r \in Q, \index_\sigma(n_i)(r) = k\}$.
Then 
By definition of $R^\infty(b)$ this $\sigma$-play has limsup $q$.
\end{proof}



%

\section{Existence of everywhere thick subtrees  (Proof of  Lemma~\ref{lem:thick})}

\begin{proof}[Proof of Lemma~\ref{lem:thick}]
The uniform measure 
$\mu$ on $\{0,1\}^\omega$ is Borel with respect to the topology whose cylinders are the basis. This topology is metrizable thus $\mu$ is inner-regular (cf.~\cite[Theorem 17.10]{kechris}).
Hence $P$ contains a closed set such that $\mu(P)>0$,
w.l.o.g. we assume that $P$ itself is closed,
i.e. $P$ contains every branch whose every node is visited by a branch in $P$.

Given a node $n\in\{0,1\}^*$ we denote $P_n$ the set of branches of $P$ visiting $n$, i.e. $P_n = P \cap (n\{0,1\}^\omega)$.
Let $T$ be the set of nodes $n$ such that $\mu(P_n)>0$.
By hypothesis $T$ contains $\epsilon$. And $T$ is prefix-closed thus it is a tree with root $\epsilon$. 
And $\vec{T}\subseteq P$ because $P$ is closed.
We show that $T$ is everywhere thick. Let $n \in T$. 
By definition, $\vec{T} \cap (n\{0,1\}^\omega) = P_n \setminus 
\bigcup_{n' \not \in T} P_{n'}$. By definition of $T$, for every $n' \not \in T$, $\mu(P_{n'})=0$ thus since $T$ is countable,
$\mu(\bigcup_{n' \not \in T} P_{n'})=0$ hence
$\mu(\vec{T} \cap (n\{0,1\}^\omega))=\mu(P_n)>0$ since $n\in T$.
\end{proof}

\section{Characterization of positively winning strategies (Proof of  Lemma~\ref{lem:caracnonzero})}

For commodity, we recall the definition of a positive witness.

Let $Z$ be the set of $\sigma$-reachable vertices whose state is in $\Qzero$.
A positive witness for $\sigma$ is a pair $(W,E)$
where:
\begin{align*}
 &W \subseteq Z \text{ are the \emph{active} vertices},\\
&E \subseteq \{0,1\}^* \times \{0,1\}\text{ is the set of \emph{positive edges},}
\enspace
\end{align*}
and they have the following properties.
\begin{itemize}
\item[a)]
From every vertex $z \in Z$ there is a positive and canonical finite $\sigma$-play starting in $z$ which reaches a vertex in $W$ or a transient vertex.
\item[b)]
Let $z=(n,q) \in W$.
Then $(n,0)\in E$ or $(n,1)\in E$, or both.
If $z \to z'$ is a local transition
then $z' \in W$ as well
whenever ($q\in Q_E$ and $z \to z'$ is consistent with $\sigma$)
or ($q\in Q_A$ and
$z\to z'$ is canonical).
If $z$ is controlled by Eve and $\sigma(z)$ is a split transition
 $q \to (q_0,q_1)$ then
$
((n,0)\in E \implies (n0,q_0) \in W)$
and
$((n,1)\in E \implies (n1,q_1) \in W)$.
 \item[c)]
 The set of nodes $\{ nd\in\{0,1\}^* \mid (n,d) \in E \}$
 is everywhere thick.
 \end{itemize}

{
{\bf We first show that existence of a positive witness is a sufficient condition for $\sigma$ to be positively winning}.
Let $(W, E)$ be a positive witness for $\sigma$.
Let $\tau$ be any strategy for Adam
and $\pi$ be a play consistent with $\sigma$ and $\tau$ whose last vertex
$z=(n,q)$ belongs to $Z$.
We show that the set $X$ of positive continuations of $\pi$ consistent with $\sigma$ and $\tau$
has nonzero probability under $\proba^{\sigma,\tau}$.
We prove this by induction on $q$ for the topological order of $Q$
in the transition graph of the automaton.
If one of the continuations in $X$ reaches a vertex $(n',q')$ with $q' < q$ then we conclude by inductive hypothesis.
In the remaining case, note that 
\begin{itemize}
\item[(*)] in all positive continuations of $\pi$ consistent with $\sigma$ and $\tau$, the strategy $\tau$ only plays canonical moves.
\end{itemize}
Let $\pi'$ be the positive play whose existence is given by property a) in the definition of positive witnesses.
Since $\pi'$ is canonical, then according to (*), $\pi\pi'$ is consistent with $\tau$ thus $\pi\pi'$ belongs to $X$.
Thus according to (*) again, $\pi'$ does not reach any transient vertex,
hence the last vertex $(n',q')$ of $\pi'$ belongs to $W$.
Let $N_E=\{ nd\in\{0,1\}^* \mid (n,d) \in E \}$.
Let $Y$ be the set of positive continuations of $\pi\pi'$ consistent with $\sigma$ and $\tau$ and staying in $N_E\times \Qzero$.
Since $Y\subseteq X$ then according to (*) all plays in $Y$
are canonical .
Thus according to property b) in the definition of nonzero-witnesses,
all plays in $Y$ stay in $W$ after $\pi\pi'$.
According to b), there exists $d \in \{0,1\}$ such that 
$n'd\in N_E$
and according to c) the set $N_E$ set is everywhere thick.
The projection of $Y$ on $\{0,1\}^\omega$ contains all branches of the subtree of $N_E$ rooted at $n'd$.
Since $N_E$ is everywhere thick,
the set of branches of this subtree has positive measure hence $\Proba^{\sigma,\tau}(Y)>0$ according to Lemma~\ref{probmu}.
}

\medskip 

{\bf Now we show the condition is necessary.}
Let $Z'$ be the subset of vertices in $Z$ from which no canonical positive $\sigma$-play leads to a transient vertex.

If $Z'$ is empty, then $(\emptyset,\{0,1\}^*\times\{0,1\})$ is a positive witness which concludes the proof. 

Otherwise $Z'$ is infinite (because of the normalization properties (N1) and (N2)).
Let $((n_k,q_k))_{k \in \NN}$ be a bread-first enumeration of all nodes in  $Z'$.
For $k\in \NN$, we set $z_k=(n_k,q_k)$.

Denote $R_k$ the set of vertices reachable from $z_k$ 
by a positive and canonical $\sigma$-play and $T'_k$
the set of nodes of these vertices.
By definition of $Z'$, once a $\sigma$-play has visited $z_k\in Z'$,
as long as Eve plays $\sigma$ and Adam plays canonical moves 
then no transient vertex is visited thus Adam always has a canonical choice when he has to take a decision. Thus, since $\sigma$ is positively winning, the set of canonical positive 
$\sigma$-play starting in $z_k$ has nonzero probability. And this is true from every vertex visited from one of these plays. 
Thus $T'_k$ is an everywhere thick subtree.

We are going to combine the vertices and nodes of $(R_k,T'_k)_{k\in \NN}$ in order
to define inductively, for every $k\in \NN$
an integer $n_k$ and a collection $\mathcal{C}_k=\left(T_{i,k}, X_{i,k}\right)_{i \in 1 \ldots n_k}$ 
of sets of nodes and vertices
with the following properties.
\begin{enumerate}[i)]
\item
For every $i\in 1\ldots n_{k}$, the set of nodes $T_{i,k}$
is an everywhere thick subtree,
whose root is denoted $r_{i,k}$.
And $X_{i,k}$ is a set of vertices at the root
i.e. $X_{i,k}\subseteq  \{r_{i,k}\}\times Q$
\item
All trees $(T_{i,k})_{i\in 1 \ldots n_k}$ are disjoint.
\item
For every $i \in  1\ldots n_{k}$ and
vertex $x\in X_{i,k}$ 
denote $W_{x,k}$
the set of vertices reachable by
canonical $\sigma$-plays starting from $x$
and visiting only nodes in $T_{i,k}$.
(The notation $W_{x,k}$ is unambigous because according to ii) there is a unique possible $i$ given $x$ and $k$).
Denote
$
X^{(k)} =
\bigcup_{i \in  1\ldots n_{k}} X_{i,k} 
$
and
$
W^{(k)}
=
\bigcup_{x\in X^{(k)}} 
W_{x,k}\enspace.
$
Then all vertices in $W^{(k)}$ are positive (i.e. $W^{(k)}\subseteq \{0,1\}^* \times \Qzero$) and none of them are transient.
\item 
The sets in the collection $(W_{x,k})_{x\in X^{(k)}}$
are disjoint i.e. this collection is a partition of $W^{(k)}$.
\item
For every $j \in 0 \ldots k$,
\be\label{ddff}
W^{(k)}
\text{ is reachable from $z_j$ by a canonical positive $\sigma$-play.}
\ee
 \end{enumerate}

\medskip

Initially we set $n_0=1$ and $T_{0,0}=T'_0$ and $X_{0,0}=\{z_0\}$.
Property v) holds since obviously $z_0\in W_{z_0,0}$,
property iii) holds since the vertices in $W_{z_0,0}$
are exactly those visited by positive and canonical $\sigma$-plays starting in $z_0$. 
Properties i) ii) and iv) are trivial since $T'_0$ is everywhere thick, $n_0=1$ and $X_{0,0}$ is a singleton.

\smallskip

We assume $k>0$
and perform the inductive definition of $\mathcal{C}_{k}=(T_{i,k},X_{i,k})_{1 \ldots n_k}$ from $\mathcal{C}_{k-1}$.
It is split in three cases: the copy case, the expansion case 
and shrinking case.

\begin{itemize}
\item {\bf Copy case.}
Assume first that $W^{(k-1)}$ is reachable from $z_{k}$ by a positive canonical $\sigma$-play.
Then $\mathcal{C}_{k}$
is simply the copy of $\mathcal{C}_{k-1}$.
Properties i)-iv) are maintained since they only depend on the collection $\mathcal{C}_{k-1}$, independently of $k$.
And v) is maintained by inductive hypothesis
for $j\in 0\ldots k-1$ and by hypothesis for $j=k$.
\item {\bf Expansion case.}
We consider $T'_k$,
the set of nodes visited
by canonical positive plays starting in $z_k$,
which is an everywhere thick subtree as discussed previously.
The expansion case occurs when
for every $ i \leq n_{k-1}$ the set of branches of the (possibly empty) subtree $T'_k \cap T_{i,k-1}$ has probability $0$. Then $T'_k \setminus \bigcup_{i \leq n_{k-1}} T_{i,k-1}$ has the same measure than $T'_k$, which is positive, thus
it contains some everywhere thick subtree $T''$
(cf. Lemma~\ref{lem:thick}).
By definition of $T'_k$
there exists a vertex $z$ at the root of $T''$ reachable from $z_k$
by a positive canonical $\sigma$-play.
The collection $\mathcal{C}_{k}$
is obtained by  adding to $\mathcal{C}_{k-1}$
the new entry $(T'',\{ z \})$ 
at index $n_k=n_{k-1}+1$.
Properties i)-v) are inherited from the inductive hypothesis,
for what concerns $\mathcal{C}_{k-1}$
and $j\in 0\ldots k-1$ for item v).
Properties i) and ii) are clear for the new entry $(T'',\{ z \})$
at $i=n_k$. Property iii) holds for $i=n_k$ by definition of $T'_k$ and because $z_k\in Z'$. Property iv) holds because of ii) and $\{z\}$ is a singleton. Property v) holds for $j=k$ by choice of $z$.
\item {\bf Shrinking case.}
We are left with the case 
where
 $W^{(k-1)}$ is \emph{not} reachable from $z_{k}$ by a positive canonical $\sigma$-play
and
there is $i\in 0\ldots n_{k-1}$
such that the set of branches of the subtree $T'_k \cap T_{i,k-1}$
has positive probability.
Then $n_k=n_{k-1}$ and $\mathcal{C}_{k}$
is the copy of $\mathcal{C}_{k-1}$
except at rank $i$, where we replace
$T_{i,k-1}$ and $X_{i,k-1}$
by 
$T_{i,k}\subseteq T_{i,k-1}$ 
and $X_{i,k}$ defined as follows.
The tree $T_{i,k}$ is set to be any everywhere thick subtree 
 contained in $T'_k \cap T_{i,k-1}$, whose existence is given by Lemma~\ref{lem:thick}.
For every vertex $x\in X_{i,k-1}\cup \{z_k\}$,
the node of $x$ is
either the root of $T_{i,k-1}$
(when $x \in X_{i,k-1}$)
 or the root of $T'_k$
 (when $x=z_k$).
 Since $T_{i,k-1}\subseteq T'_k \cap T_{i,k-1}$,
the node of $x$ is an ancestor of the root of $T_{i,k}$.
We show 
that
 there is a positive canonical $\sigma$-play
$\pi_x$ from $x$ to a vertex on the root of $T_{i,k}$.
In case $x \in X_{i,k-1}$ because
$T_{i,k}\subseteq T_{i,k-1}$
and property iii).
In case $x=z_k$ because
$T_{i,k}\subseteq T_{i,k-1}$ and by definition of $T'_k$.
Denote $w(x)=\last(\pi_x)$.
We set
\[
X_{i,k} =  \left\{ w(x) \mid x\in X_{i,k-1}\cup \{z_k\}\right\}\enspace.
\]
Remark that
\begin{equation}
|X_{i,k}| = |X_{i,k-1}| + 1\enspace\label{eq:ind}.
\end{equation}
The reason for~\eqref{eq:ind} is that, according to iv),
all the vertices $(w(x))_{x \in X_{i,k-1}}$
are distinct.
And since $\{ w(x) \mid x \in X_{i,k-1}\}\subseteq W^{(k-1)}$
and since we are not in the copy case
then $w(z_k) \not \in \{ w(x) \mid x \in X_{i,k-1}\}$.

The construction of $\mathcal{C}_{k}$
 preserves invariants i)
and ii)
because $(T_{i,k})_{i\in 1 \ldots n_k}$
is the copy of 
$(T_{i,k-1})_{i\in 1 \ldots n_{k-1}}$
except for the tree
$T_{i,k}$ which is everywhere thick by construction
and contained in $T_{i,k-1}$.
Invariant iii) is preserved for  $x=w(x'), x' \in X_{i,k-1}$
since $w(x)$
is reachable from $X^{(k-1)}$ by a canonical positive
$\sigma$-play thus
$W_{w(x),k}\subseteq W_{w(x),k-1}$
Invariant iii) is true for $x=w(z_k)$ because
 $T_{i,k} \subseteq T'_k$
 thus from $z_k$ every canonical $\sigma$-plays is
 positive and since $z_k\in Z'$ then no transient vertex is 
 reached by such a play.
 Invariant iv) is true by hypothesis on $z_k$.
 Invariant v) is preserved for $j\in 0\ldots k-1$ because from every vertex in $W^{(k-1)}$ there is a positive canonical $\sigma$-play
 to a vertex in $W^{(k)}$.
 And invariant v) is true for $j=k$ because $z_k \in X_{i,k}\subseteq X^{(k)} \subseteq W^{(k)}$. 
\end{itemize}

Let $I$ be the set of values taken by the sequence
$1=n_0\leq n_1\leq \ldots$.
According to~\eqref{eq:ind},
the shrinking case can occur at most $|\Q|$ times at the same index $i\in I$ (in the copy and expansion cases $X_{i,k}$ is not modified).
Thus, if $k$ is the smallest rank at which $n_k=i$,
the family 
$(T_{i,k'}, X_{i,k'})_{k' \geq k}$ takes at most $|\Q|$ different values and is ultimately constant from some rank $k_i$.

We are now ready to define the positive witness $(E,W)$
using the family $(T_{i,k_i}, X_{i,k_i})_{i \in I}$  obtained "at the limit".
Let 
$E$ be the set of edges appearing in one of the trees
of the family $(T_{i,k_i})_{i \in I}$ i.e.
\[
E = \bigcup_{i \in I} \left \{ (n,b) \mid  (n \in T_{i,k_i}) \land b \in \{0,1\} \land (nb \in T_{i,k_i}) \right\}
\]
and 
\[
W = \bigcup_{i \in I, x \in X_{i,k_i}} W_{x,k_i}\enspace.
\]
According to invariant iii), $W$ is a set of positive vertices.
To prove that $(W,E)$ is a positive witness we should establish properties a) b) and c).
Property a) is clear for vertices $z \in Z \setminus Z'$,
because by definition of $Z'$ from these vertices there is a positive canonical $\sigma$-play to a transient vertex. And if $z \in Z'$ then $z=z_k$ for some $k$ and
a) is a consequence of v).

We prove that property b) holds.
Let $z=(n,q) \in W$ then there is $i\in I$ and $x\in X_{i,k_i}$
such that $z\in W_{x,k_i}$.
Then $n\in T_{i,k_i}$,
by definition of the sets $W_{x,k_i}$.
Since $T_{i,k_i}$ is everywhere thick then $n$ is not a leaf
of $T_{i,k_i}$
thus there exists $d\in\{0,1\}$ such that $nd\in T_{i,k_i}$
hence $(n,d)\in E$.
Let $z \to z'$ a local transition like in property b). Then,
by definition of $W_{x,k_i}$, 
$z'\in W_{x,k_i}$ thus $z'\in W$.
Assume Eve controls $q$ and $\sigma(z)$ is a split transition $(q_0,q_1)$ and let $d'$ such that $(n,nd')\in E$.
Since $x\in T_{i,k_i}$
then according to ii) and iv), 
also $n\in T_{i,k_i}$
thus by definition of $E$ also $n'\in T_{i,k_i}$.
Then, by definition of $W_{x,k_i}$,
$z'\in W_{x,k_i}$ thus $z'\in W$.

Property c) holds.
We show that $N_E=\{ nd\in\{0,1\}^* \mid (n,d) \in E \}$
is everywhere thick.
Let $n\in N_E$.
By definition of $E$,
there exists $i\in I$ such that $n\in T_{k_i}$
thus 
$\vec{T_{k_i}} \cap n\{0,1\}^\omega\subseteq \vec{N_E} \cap n\{0,1\}^\omega$
(actually this is an equality according to ii))
and $\mu(\vec{N_E} \cap n\{0,1\}^\omega)\geq \mu(\vec{T_{k_i}} \cap n\{0,1\}^\omega)>0$
since $T_{k_i}$ is everywhere thick.

\section{From \SP\ to \BSP: proof of Theorem~\ref{theo:reduc}}
Theorem~\ref{theo:reduc} is a corollary of Lemma~\ref{lem:sateq}
at the end of the section. 

Before translating formulas,
we turn them in positive form,
where the only negations are in front of letters of the alphabet.

\begin{lemma}\label{lem:posform}
Given a formula $\pCTLFormula$ one can build an equivalent formula in positive form whose size (as a DAG) is linear in the size of the first formula.
\end{lemma}
\begin{proof}
These transformations
 preserve the size of the DAG of
the formula and its models:
\begin{align*}
 \neg\neg \pf&\to\pf
& \neg \proba_{>0}( \pf) &\to \proba_{=1}( \neg\pf)
& \neg \top &\to \bot \\
\neg \proba_{=1}( \pf) &\to \proba_{>0}(\neg \pf)
& \neg \next \pf &\to (\next \neg \pf)
& \neg \exists \pf &\to \forall \neg \pf\\
\neg \always \pf &\to (\top \until  \neg \pf)
&\neg \forall \pf &\to \exists \neg \pf
&\neg (\pf_1 \until  \pf_2)
&\to
 \always \neg \pf_2 \lor (\neg \pf_2 \until \neg \pf_1)\enspace.
\end{align*}
\end{proof}

\paragraph*{From Markov chains to binary trees: adding the delay symbol $\blank$}
There is a natural transformation of a Markov chain into a binary  tree,
which preserves the probability measure on $\Sigma^\omega$, up to some projection.
This transformation simulates a single transition of the original Markov chain by an unbounded number of transitions of the binary  tree.
This requires to introduce in the alphabet of the binary  tree,
on top of the alphabet $\Sigma$ of the Markov chain,
a new \emph{delay} symbol $\blank$.
In the sequel we use the notation
\[
\Sigma_\blank = \Sigma \cup \{\blank\}\enspace.
\]

Every infinite path in the Markov chain labelled by
a word $u=a_0a_1a_2\cdots \in \Sigma^\omega$
corresponds to some infinite branch of the tree
labelled by a word in $a_0\{\blank\}^* a_1\{\blank\}^* a_2\cdots \in (\Sigma\{\blank\}^*)^\omega$

Intuitively, the symbol $\blank$ delays the stochastic process,
but it should not delay it forever and is expected to
appear a finite number of times between two occurences of a non-delay symbol in $\Sigma$,
thus we are interested in $\Sigma_\blank$-labelled Markov chains with finite delay.

\begin{definition}[Finite delay]
A $\Sigma_\blank$-labelled
Markov chain $\mc_\blank$ has \emph{finite delay}
if from every state $s_\blank$ there is probability $0$
to see the delay symbol forever:
$\proba_{\mc_\blank,s_\blank}(\{\blank\}^\omega)=0$.
\end{definition}


A $\Sigma_\blank$-labelled chain with finite delay is almost a $\Sigma$-labelled chain,
up to a \emph{projection}.
%
\begin{lemma}[Projecting Markov chains with finite delay]
\label{lem:fdcarac}
A $\Sigma_\blank$-labelled Markov chain
$\mc_\blank=(\states_\blank,\lab_\blank,\pt_\blank)$
has \emph{finite delay}
iff 
there exists a $\Sigma$-labelled Markov chain $\mc=(\states, \lab,\pt)$
and 
\[
\pi: \{ s_\blank \in \states_\blank \mid \lab_\blank(s_\blank)\neq \blank\} \to \states
\text{ such that:}
\]
\begin{itemize}
\item
$\pi$ is surjective; and
\item
for every state $s_\blank\in\states_\blank$,
$
(\lab_\blank(s_\blank)\neq\blank)
\implies
(\lab_\blank(s_\blank)=
\lab(\pi(s_\blank)))$; and
\item
for every state $s,u\in\states$ and $s_\blank\in\states_\blank$ such that $\pi(s_\blank)=s$
\begin{align}
\pt(s,u)
=
\proba_{\mc_\blank,s_\blank}(\{
s_\blank s_1\cdots s_{n-1}u_\blank \in \states_\blank^*
\mid
\lab_\blank(s_1)=\cdots = \lab_\blank(s_{n-1})=\blank, \pi(u_\blank)=u\}
)\enspace.
\label{eq:corr}
\end{align}
\end{itemize}
Such a map $\pi$ is called a \emph{projection} of $\mc_\blank$ to $\mc$.
\end{lemma}

Before giving the proof, we start with a preliminary lemma:
Actually projections preserve probability measures and path labelling.
Denote
$\pathes^{r}_{\mc_\blank}$
the set of pathes of $\mc_\blank$ that go through infinitely many $\Sigma$-labelled vertices.
The definition of $\pi$ is extended to
$
\pi : \pathes^{r}_{\mc_\blank} \to \pathes_{\mc} \enspace,
$
by erasing states with label $\blank$
and projecting $\Sigma$-labelled states to their image by $\pi$.
Then,

\begin{lemma}\label{lem:projprop}
For every state $s_\blank\in\lab_\blank^{-1}(\Sigma)$
and every measurable set $E \subseteq \states^\omega$,
\begin{align}
\label{eq:eq}
&\proba_{\mc,\pi(s_\blank)}(E)
= \proba_{\mc_\blank,s_\blank}(\pi^{-1}(E))
 \enspace,\\
\label{eq:eq2}
&\pathes_{\mc}(\pi(s_\blank))
=
\pi\left(\pathes^{r}_{\mc_\blank}(s_\blank)\right)
\enspace.
\end{align}
\end{lemma}
\begin{proof}
By definition of projections,
property~\eqref{eq:eq} holds when $E$ is a cylinder of length $1$ i.e. $E=st\states^\omega$.
An easy induction show that it also holds when $E$ is any cylinder $E=ss_0\cdots s_n\states^\omega$.
Since property~\eqref{eq:eq} is stable by complement and countable unions, it holds for every measurable set $E$.

Property~\eqref{eq:eq2}
holds because
for every $n\in\NN$,
 the finite pathes of length $n$ in $\mc$
 are exactly the projection by $\pi$ of pathes in $\mc_\blank$ that go through exactly
 $n$ $\Sigma$-labelled vertices.
\end{proof}

\begin{proof}[Proof of Lemma~\ref{lem:fdcarac}]
The finite delay hypothesis is necessary.
Let $\pi$ be a projection from $\mc_\blank$ to $\mc$
and $s_\blank$ a state of $\mc_\blank$ then
\begin{align*}
\proba_{\mc_\blank,s_\blank}(\Sigma\{\blank\}^\omega)
&=
1 - \proba_{\mc_\blank,s_\blank}(\{s_\blank s_1\cdots s_n \mid \lab_\blank(s_n)\neq \blank\})\\
&=
1 - \sum_{u\in \states}p(\pi(s_\blank),u))
\text{ (according to~\eqref{eq:corr})}\\
&=0\text{ (since $\mc$ is a Markov chain).}
\end{align*}

The finite delay property is sufficient.
The projection $\mc=(\states,\lab,\pt)$
of $\mc_\blank=(\states_\blank,\lab_\blank,\pt_\blank)$ is defined by
$
\states = \{ s \in \states_\blank\mid \lab_\blank(s) \neq \blank \}
$,
$\lab$ is the restriction of $\lab_\blank$ on $\states$
and for every states $s,u\in\states$,
\[
\pt(s,u)=\proba_{\mc_\blank,s}(ss_1\cdots s_nu \mid \lab_\blank(s_1)=\ldots=\lab_\blank(s_n)=\blank\}\enspace.
\]
The projection $\pi$ is the identity on $\states$.
The finite delay hypothesis guarantees that $\forall s \in \states,\sum_{u} \pt(s,u)=1$ thus $\mc$ is a Markov chain.
\end{proof}

\paragraph*{Equivalence between Markov chain and binary trees}


\begin{lemma}[Every chain is the projection of a binary tree]
\label{lem:chaintotree}
Let $\mc=(\states, \lab,p)$
be a $\Sigma$-labelled Markov chain
and $s_0$ a state of $\mc$.
Then there exists a $\Sigma_\blank$-labelled binary tree
$ \lab_\blank:\{0,1\}^* \to \Sigma_\blank$
and a projection $\pi$ of $\lab_\blank$ to $\mc$
such that $\pi(\epsilon)=s_0$.
\end{lemma}

\begin{proof}[Proof of Lemma~\ref{lem:chaintotree}]
The tree $\lab_\blank$
and
the projection mapping
\[
\pi: \{ w \in \{0,1\}^* \mid \lab(w) \neq \blank \} \to \states
\]
are defined inductively.
Initially we set
$
\pi(\epsilon) = s_0
$ and $\lab_\blank(\epsilon)=\lab(s_0)$.
Assume that $\pi$ and $\lab_\blank$ are already defined for some node $w \in \{0,1\}^*$
such that $\pi(w)\in\states$.
Denote $s=\pi(w)$.
We fix an enumeration (finite or infinite) of the successors of $s$ in $\mc$.
\[
\{z_1,z_2,\ldots \} = \{ z \in S \mid \pt(s,z) > 0\}\enspace.
\]

For every $z_i$ we are going to define a subtree $T_i$
rooted on $w$ such that
the leaves of $T_i$ are mapped by $\pi$ to $z_i$
(and thus are labelled by $\lab(z_i)$)
and the inner nodes of $T_i$ are labelled by $\blank$.

The construction of $T_i$ makes use of the usual continous mapping 
$\phi : \{0,1\}^* \to [0,1]$ which associates with every finite sequence of bits  $\delta_1\cdots  \delta_n\in\{0,1\}^*$
the real number
\[
\phi(\delta_1\cdots \delta_n) = \sum_{ 1 \leq k \leq n } \frac{\delta_i}{2^i} \in [0,1[\enspace.
\]
We set $p_0=0$ and for every successor $z_i$ of $s$:
\begin{align*}
&
p_i = p(s,z_1) + \ldots + p(s,z_{i})\\
&T_i
=
\left\{
\delta_1\cdots \delta_n\in\{0,1\}^+
\mid
\phi(\delta_1 \cdots \delta_n\{0,1\}^*) \subseteq ]p_{i-1}, p_i [
\right\}\\
&L_i = \{ \delta_1\cdots \delta_n \in T_i \mid \delta_1\cdots \delta_{n-1} \not \in T_i\}
\enspace.
\end{align*}
The definition of $\pi$ is expanded to $\bigcup_{z_i} wL_i$, as follows.
\begin{equation}
\label{eq:sum}
\forall \delta_1\cdots \delta_n \in L_i,
\pi(w\delta_1\cdots \delta_n) = z_i\enspace.
\end{equation}

%
The condition~\eqref{eq:corr} in the definition of a projection holds:
denote $\mu$ the uniform Lebesgue measure on $[0,1]$, then
for every successor $z_i$,
\begin{align*}
&\proba_{\lab_\blank,w}(\{w, w_1,\cdots, w_n \in \states_\blank^* \mid
\lab_\blank(w_1)=\ldots=\lab_\blank(w_{n-1})=\blank \text{ and } \pi(w_n)=z_i \})\\
&
=\sum_{\delta_1\cdots \delta_n \in L_i } \proba_{\lab_\blank,w}( \text{reach node } w\delta_1\cdots \delta_n)
= \sum_{\delta_1\cdots \delta_n \in L_i} \frac{1}{2^n}\\
&
= \sum_{\delta_1\cdots \delta_n \in L_i} \mu(\phi(\delta_1\cdots \delta_n\{0,1\}^*))\\
&
= \mu(]p_{i-1},p_i[)
=p_i - p_{i-1} \\
&= p(s,z_i)\enspace.
\end{align*}
The first equality is by inductive definition of $\pi$ and $\lab_\blank$,
the second is by definition of a Markov binary tree,
the third is a simple computation, as well as the two last ones.
The fourth equality holds because the collection of intervals
$
\left(\phi(\delta_1\cdots \delta_n\{0,1\}^*)\right)_{ \delta_1\cdots \delta_n \in L_i}
$
is a partition of $]p_{i-1},p_i[$: by definition these intervals are contained
in $]p_{i-1},p_i[$, they are disjoint because $\phi(w\{0,1\}^*)\cap \phi(w'\{0,1\}^*)\neq \emptyset$
implies that $w \sqsubseteq w'$ or $w' \sqsubseteq w$
but $L_i$ is prefix-free
and for any $x \in ]p_{i-1},p_i[$
there exists $k,n$ such that
\[
x \in
\left [\frac{k}{2^n}-\frac{1}{2^n}, \frac{k}{2^n}+\frac{1}{2^n}\right [
\subseteq  \left ] p_{i-1},p_i \right [\enspace.
\]
This terminates the inductive step of the construction of $\pi$ and $\lab_\blank$ and the proof of the lemma.
\end{proof}

The correspondance between $\Sigma_\blank$-labelled Markov chains and their projections
on $\Sigma$ has a logical counterpart.
For every \ctlsa\  formula $\fo$ on the alphabet $\Sigma$,
there is a similar \ctlsa\ formula $\ov{\fo}$
on the alphabet $\Sigma_\blank$,
such that a $\Sigma_\blank$-labelled Markov
chain satisfies  $\ov{\fo}$ if and only if its
projection on $\Sigma$ satisfies $\fo$.

\begin{definition}[Lifting of a \ctlsa\  formula on $\Sigma_\blank$]
Let $\fo$ be a \ctlsa\  formula on $\Sigma$.
The lifting of $\fo$ on $\Sigma_\blank$
is the \ctlsa\ formula $\ov{\fo}$ defined inductively by
\begin{align*}
 &\ov{\next \pf}=
\next (\blank \until \ov{\pf})
&&\ov{\exists\pf} =  (\neg \blank) \land\exists \left(\ov{\pf} \land \neg (\top \until G\blank) \right)
\\
&\ov{\always \pf} =
 \always (\blank \lor  \ov{\pf})
&&\ov{\forall\pf} =  (\neg \blank) \land\forall \left(\ov{\pf} \lor (\top \until G\blank)\right)
\\
&\ov{\pf_1 \until \pf_2} =
 (\blank \lor  \ov{\pf_1}) \until \ov{\pf_2} 
&&\ov{\proba_{\sim b}(\pf)} =  (\neg \blank) \land
\proba_{\sim b}(\ov{\pf})
\end{align*}
and the transformation  from $\fo$ to $\ov{\fo}$
leaves other operators invariant.
\end{definition}

\begin{lemma}[Lifting are compatible with projections]
\label{lem:logeq}
Let $\pi$ a projection of a
$\Sigma_\blank$-labelled Markov chain $\mc_\blank$
on a $\Sigma$-labelled Markov chain $\mc$.
For every state-formula $\sf$,
and for every $\Sigma$-labelled state
$s_\blank$ of $\mc_\blank$,
\[
(\mc_\blank, s_\blank \models \ov{\sf})\iff
(\mc, \pi(s_\blank) \models \sf)\enspace.
\]
\end{lemma}
\begin{proof}[Proof of Lemma~\ref{lem:logeq}]
The equivalence
\be
\label{eq:liftstate}
(\mc_\blank, s_\blank \models \ov{\sf})\iff
(\mc, \pi(s_\blank) \models \sf)\enspace.
\ee
is proved by induction on $\sf$,
together with the following extra property~\eqref{eq:liftpath}.
Denote
\[
\pathes^{r}_{\mc_\blank} =  \pathes_{\mc_\blank} \setminus \left((\Sigma\cup\{\blank\})^*\{\blank\}^\omega\right)
\]
the set of pathes of $\mc_\blank$ that go through infinitely many $\Sigma$-labelled vertices.
The definition domain of $\pi$ is extended to
\[
\pi : \pathes^{r}_{\mc_\blank} \to \pathes_{\mc} \enspace.
\]
by projecting $\Sigma$-labelled states $s$ to $\pi(s)$ and $\blank$-labelled states to the empty word $\epsilon$.
We show by induction that for every path formula $\pf$,
and every path $\aPath \in \pathes^{r}_{\mc_\blank} $
\begin{align}
\label{eq:liftpath}
(\mc_\blank,w \models \ov{\pf}) \iff (\mc,\pi(w) \models \pf)\enspace.
\end{align}

\medskip

Every state formula $\sf$ is also a path formula.
Assume that property~\eqref{eq:liftstate} holds for $\sf$
and some state $s$.
Then property~\eqref{eq:liftpath} holds for $\sf$ (seen as a path formula)
and every path starting from $s$, by definition of $\aPath \models \sf$.

Assume that property~\eqref{eq:liftpath} holds for some path formula $\{\pf',\pf_1,\pf_2\}$.
We show that~\eqref{eq:liftpath}
holds for
$\pf\in\{\next \pf',\always \pf',  \pf_1
\until \pf_2\}$.
Let $\aPath \in \pathes^{r}_{\mc_\blank}$,
Since $\aPath$ visits infinitely many $\Sigma$-labelled states,
there is an infinite sequence of integers $n_0, n_1, n_2,\ldots\in \nats$
such that
 $\lab_\blank(w)= a_0\blank^{n_0} a_1 \blank^{n_1}a_2\cdots$.
 Then by definition of $\pi$,  $\lab(\pi(w))=a_0a_1a_2\cdots$.
We set $m_0=0 \leq m_1=1+n_1 \leq m_2=1+n_1+1+n_2 \leq \ldots$
the positions where the symbols in $\Sigma$ appear in $w$.
Then
\be\label{eq:mi}
\forall i \in \NN, \pi(w[m_i]) = \pi(w)[i]\enspace.
\ee

The equivalence~\eqref{eq:liftpath} holds for $\pf=\next \pf'$ because
\begin{align*}
(\mc_\blank,w \models  \ov{\next\pf'})
&\iff
(\mc_\blank,w \models \next (\blank \until \ov{\pf'}))
& (\text{by definition of $\ov{\pf'}$})\\
&\iff
(\mc_\blank,w[m_1] \models \ov{\pf'})
& (\text{by definition of $m_1$})
\\
&\iff
(\mc,\pi(w[m_1]) \models \pf')
& (\text{by induction hypothesis})
\\
&\iff
(\mc,\pi(w)[1] \models \pf')
& (\text{according to~\eqref{eq:mi}})\\
&\iff
(\mc,\pi(w) \models\next \pf')
& (\text{by definition of $\models\next \pf'$})
\enspace.
\end{align*}
With the same arguments, the equivalence~\eqref{eq:liftpath} holds for $\pf=\always \pf'$ because
\begin{align*}
(\mc_\blank,w \models  \ov{\always\pf'})
&\iff
(\mc_\blank,w \models \always (\blank \lor \ov{\pf'}))\\
&\iff
(\forall i\geq 0, \mc_\blank,w[m_i] \models \ov{\pf'})\\
&\iff
(\forall i\geq 0, \mc,\pi(w[m_i]) \models {\pf'})\\
&\iff
(\forall i\geq 0, \mc,\pi(w)[i]) \models {\pf'})\\
&\iff
(\mc,\pi(w) \models \always \pf')
\enspace.
\end{align*}
And the equivalence~\eqref{eq:liftpath} holds for
$\pf=\pf_1\until \pf_2$ because
\begin{align*}
&(\mc_\blank,w \models  \ov{\pf_1\until \pf_2})\\
&\iff
(\mc_\blank,w \models (\blank \lor  \ov{\pf_1}) \until \ov{\pf_2})\\
&\iff
\exists j \geq 0, (\mc_\blank,w[m_j]\models\ov{\pf_2})
\land  \forall 0\leq i < j, (\mc_\blank,w[m_i]\models\ov{\pf_1})\\
&\iff
\exists j \geq 0, (\mc,\pi(w[m_j])\models \pf_2)
\land  \forall 0\leq i < j, (\mc,\pi(w[m_i])\models \pf_1)\\
&\iff
\exists j \geq 0, (\mc,\pi(w)[j])\models \pf_2)
\land  \forall 0\leq i < j, (\mc,\pi(w)[i])\models \pf_1)\\
&\iff
(\mc,\pi(w) \models \pf_1 \until \pf_2)
\enspace.
\end{align*}
This terminates the inductive proof of~\eqref{eq:liftpath},
under the hypothesis that~\eqref{eq:liftstate} and~\eqref{eq:liftpath}
hold for sub-formula.
\medskip

Now we show that~\eqref{eq:liftstate} holds
 $\sf\in\{\exists \pf',\forall \pf', \proba_{\sim b} \pf'\}$
whenever property~\eqref{eq:liftpath}
holds for $\pf'$.
There are three cases.
In case $\sf=\exists \pf'$,
\begin{align*}
&(\mc_\blank,s_\blank\models \ov{\exists \pf'})\\
&\iff
(\mc_\blank,s_\blank \models
(\neg \blank) \land\exists (\neg (XG\blank) \land \ov{\pf'}))
& \text{ def. of  $\ov{\exists \pf'}$}\\
&\iff
(\mc_\blank,s_\blank \models
\exists (\neg (\top \until G\blank) \land \ov{\pf'}))
& \text{ $\lab_\blank(s_\blank)\in\Sigma$}\\
&\iff
\exists \aPath_\blank \in \pathes_{\mc_\blank}(s_\blank),(\mc_\blank,
\aPath_\blank \models \neg (\top \until G\blank) \land \ov{\pf'})
& \text{ def. of $\models\exists$}
\\
&\iff
\exists \aPath_\blank \in \pathes^r_{\mc_\blank}(s_\blank),
(\mc_\blank, \aPath_\blank \models \ov{\pf'})
& \text{ by def.}
\\
&\iff
\exists \aPath_\blank \in \pathes^r_{\mc_\blank}(s_\blank),
(\mc, \pi(\aPath_\blank) \models \pf')
& \text{ ind. hyp.}
\\
&\iff
\exists \aPath \in \pathes_\mc(s),
(\mc,\aPath \models \pf')
& \text{ by~\eqref{eq:eq2} in Lemma~\ref{lem:projprop}}
\\
&\iff
(\mc,s \models \exists  \pf')
& \text{ def of $\models\exists$.}
\end{align*}
The proof of~\eqref{eq:liftstate}
in case $\sf=\forall \pf'$ is similar.
In case
$\sf=\proba_{\sim b}(\pf')$,
we denote
\begin{align*}
&L_{\ov{\pf'}}
=
\{s_\blank s_1\cdots \in \pathes_{\mc_\blank}(s_\blank) \mid s_\blank s_1\cdots\models\ov{\pf'}\}\\
&L_{\pf'}
=
\{s s_1\cdots \in \pathes_{\mc}(s) \mid s s_1\cdots\models\pf'\}
\enspace.
\end{align*}
Then according to the induction hypothesis~\eqref{eq:liftpath} for $\pf'$,
\be\label{eq:tt1}
\pi^{-1}
\left(
L_{\pf'}
\right)
=
L_{\ov{\pf'}}
~\cap~\pathes^r_{\mc_\blank}(s_\blank)
\ee
and
\begin{align*}
&(\mc_\blank,s_\blank\models \ov{\proba_{\sim b}(\pf')})\\
&\iff
(\mc_\blank,s_\blank \models
(\neg \blank) \land\proba_{\sim b}(\ov{\pf'}))
& \text{ def. of  $\ov{\proba_{\sim b}(\pf')}$}\\
&\iff
(\mc_\blank,s_\blank \models
\proba_{\sim b}(\ov{\pf'}))
& \text{ because $c_\blank(s_\blank)\in\Sigma$}\\
&\iff
\proba_{\mc_\blank,s_\blank}
\left(
L_{\ov{\pf'}}
\right)
 \sim b
& \text{ by def}\\
&\iff
\proba_{\mc_\blank,s_\blank}
\left(
L_{\ov{\pf'}}
~\cap~\pathes^r_{\mc_\blank}(s_\blank)
\right)
 \sim b
& \text{ because $\mc_\blank$ has finite delay}\\
&\iff
\proba_{\mc,s}
\left( \pi^{-1}
\left(
L_{\pf'}
\right) \right)
\sim b
& \text{ by~\eqref{eq:tt1}}\\
&\iff
\proba_{\mc,s}
\left( L_{\pf'} \right)
\sim b
& \text{ by~\eqref{eq:eq} in Lemma~\ref{lem:projprop}}\\
&\iff
(\mc,s\models
\proba_{\sim b}(\pf') )
& \text{ by def.}
\end{align*}

This terminates the inductive step,
thus property~\eqref{eq:liftpath} holds for all path formula
and property~\eqref{eq:liftstate} holds for all state formula.
\end{proof}

Finally, Theorem~\ref{theo:reduc} is a corollary of:

\begin{lemma}[Equivalence]
\label{lem:sateq}
Let $\fo$ a \ctlsa\  formula $\fo$ with alphabet $\Sigma$.
Let  $\blank$ a symbol not in $\Sigma$
and
$\ov{\fo}$
the lifting of $\fo$ on $\Sigma_\blank$.
Then the following statements are equivalent:
\begin{enumerate}
\item[i)]
$\fo$ is satisfiable.
\item[ii)]
$\ov{\fo}$ is satisfiable by a
binary tree with finite delay.
\item[iii)]
$\ov{\fo} \land \proba_{=0}(\top \until G\blank)$ is satisfiable by a 
binary tree.
\item[iv)]
$\ov{\fo} \land \proba_{=0}(\top \until G\blank)$ is satisfiable.
\end{enumerate}
\end{lemma}
\begin{proof}
We do a circular proof.
Assume that i) holds and prove ii).
Then according to i) $\fo$ is satisfiable by some
Markov chain $\mc=(\states,\lab_\mc,p)$ and state $s_0\in\states$
such that $\mc,s_0\models \fo$.
According to Lemma~\ref{lem:chaintotree},
there exists a projection $\pi$ from a $(\states \cup \{\blank\})$-labelled binary tree $\tree_\blank$
to $\mc$, such that $\pi(\epsilon)=s_0$.
According to Lemma~\ref{lem:logeq}
$\tree_\blank,\epsilon\models \ov{\fo}$,
thus ii) holds.
If ii) holds then iii) holds by definition of binary trees with finite delay.
Clearly iii) implies iv).
Assume that iv) holds and prove i).
According to iv),
$\ov{\fo}$ is satisfiable by some $(\Sigma \cup \{\blank\})$-labelled Markov chain $\mc_\blank$ with finite delay.
According to Lemma~\ref{lem:fdcarac},
there exists a projection $\pi$ of $\mc_\blank$
to a $\Sigma$-labelled
Markov chain $\mc$.
According to Lemma~\ref{lem:logeq}
$\mc,\pi(\epsilon)\models \fo$,
thus i) holds.
\end{proof}


\section{From \ctlsa\ to alternating automata: proof of Lemma~\ref{lem:pctltobc}}

\begin{proof}[Proof of Lemma~\ref{lem:pctltobc}]
Given a state formula $\sf$ in positive form
we denote $\overline \sf$ the positive form of its negation
and call it the \emph{dual} of $\sf$.

We assume that the input formula $\pCTLFormula$ is positive, which is w.l.o.g according to Lemma~\ref{lem:posform}.
Let $\mathcal{SF}$ be the set of every state 
subformula of the input formula $\pCTLFormula$, and their dual,
including $\top$ and $\bot$.
Let
$\mathcal{PF}$ be the set of path subformula
of the state formula in  $\mathcal{SF}$.

With every path formula $\pf\in \mathcal{PF}$
is associated 
a deterministic parity automaton on infinite words,
denoted $\automata_\pf$.
The alphabet of $\automata_\pf$ depends on $\pf$.
Denote $\mathcal{SF}(\pf)$ the collection
of state formulas appearing in $\pf$ i.e. $\mathcal{SF}(\pf)$ is the set of leaves of the syntactic tree of $\pf$.
Then the alphabet of $\automata_\pf$ is the collection
of subsets of $\mathcal{SF}(\pf)$.
The automaton 
$\automata_\pf$ recognizes the sequences of valuations
of the formula of $\mathcal{SF}(\pf)$ for which the path formula $\pf$ is true.
The construction of such an automaton of size $\mathcal{O}\left(2^{2^{|\phi|}}\right)$ is standard (see \emph{e.g.}~\cite{ltl}).
Note that in the case of \ECTL\ formulas, the construction is only of size $\mathcal{O}\left({2^{|\phi|}}\right)$ since path formulas are directly given as deterministic B\"uchi automata.

\paragraph*{States and transitions.}
Every state formula $\sf\in\mathcal{SF}$
is also a state of the alternating automaton.

\begin{itemize}
\item
For every letter $a\in \Sigma$, the transitions from $\sf=a$ or $\sf=\neg a$  
are local
and deterministic:
\begin{align*}
a &\to_a \top 
&
\forall b\neq a,
~~a &\to_b \bot
\\
 \neg a &\to_a \bot&
\forall b\neq a,\neg a &\to_b \top
\end{align*}
Other transitions of the automaton do not depend on the label of the current node
and are specified without mentioning the letter.
\item
From states $\sf=\sf_1\lor \sf_2$
and $\sf=\sf_1\land \sf_2$ there are local transitions
$\sf  \to \sf_1$
 and $\sf  \to \sf_2$.
%
In the $\lor$ case the choice is made by Eve
and in the $\land$ case by Adam.
\item
For every formula in $\sf\in\{\exists \pf, \forall \pf,\proba_{>0}( \pf),\proba_{=1}( \pf) \}$,
there are states $(\sf,q)_{q \in R_\pf}$,
controlled by Eve,
where $R_\pf$ is the set of states of the automaton 
$\automata_\pf$.
The state $\sf$  is the source of a unique local transition
to $(\sf,i_\pf)$, where $i_\pf$ is the initial state of the automaton 
$\automata_\pf$.
From $(\sf,q)$ with $q \in R_\pf$,
Eve can choose any subset ${\bf b}\subseteq \mathcal{SF}(\pf)$
and perform a local transition $(\sf,q)\to(\sf,q,{\bf b})$.
Intuitively, for every state formula $\sf_0$ appearing in $\pf$,
Eve has to claim whether or not this formula holds in the current node
by including or not $\sf_0$ in ${\bf b}$.
Adam controls $(\sf,q,{\bf b})$ and faces a choice.
\begin{itemize}
\item
Adam can ask for a proof of the valuation ${\bf b}$
by selecting 
a state formula $\sf_0\in \mathcal{SF}(\pf)$ and
playing the local transition $(\sf,q,{\bf b})\to\sf_0$
 if $\sf_0 \in {\bf b}$
and $(\sf,q,{\bf b})\to\overline {\sf_0}$ if $\sf_0 \not\in {\bf b}$.
\item
Adam can accept the valuation ${\bf b}$
and plays a local transition
$(\sf,q,{\bf b})\to(\sf,q',E)$ 
where $q\to_{\bf b} q'$ is the deterministic transition
of $\automata_\sf$ on letter ${\bf b}$.
From there Eve has to choose a split transition,
her options depend on the exact type of $\sf$:
\begin{itemize}
\item
If $\sf=\forall \pf$ or $\sf=\proba_{=1}( \pf)$ the only option for Eve
is the split transition to $((\sf,q'),(\sf,q'))$.
\item
If $\sf = \exists \pf$
then  Eve can choose
between two split transitions
leading  to\\
$
\text{ either }
(~\top~,~(\sf,q')~)
\text{ or }
(~(\sf,q')~,~\top~)\enspace.
$
\item
If $\sf = \proba_{>0}( \pf)$
then Eve can choose
between  three split transitions
leading to
$
\text{ either }
(~(\sf,q')~,~(\sf,q')~) \text{ or }
(~\ignore~,~(\sf,q')~) \text{ or }
(~(\sf,q')~,~\ignore~)\enspace,
$
where $\ignore$ is the special absorbing pruning state.
Also states $\top$ and $\bot$ are absorbing.
\end{itemize}
\end{itemize}
\end{itemize}

Remark that his automaton has finite choice for Adam,
the canonical choice for Adam is to accept the valuation proposed by Eve, otherwise the automaton moves to a subformula.

\paragraph*{Acceptance conditions.}
Every play of the acceptance game
either ends up in one of the three absorbing states
$\top,\ignore,\bot$ or eventually stays trapped in a $\equiv$-component 
whose all states contain the same state formula
$\sf\in\{\exists \pf, \forall \pf,\proba_{>0}( \pf),\proba_{=1}( \pf) \}$, denoted $Q_\sf$.
The acceptance conditions are defined by:

\begin{itemize}
\item[A1)]
$(\top\in\Qall\cap\Qone\cap\Qzero)~\text{and}~(\ignore\in \Qall\cap\Qone\setminus \Qzero)~\text{and}~(\bot \in \Q \setminus (\Qall\cup \Qone \cup \Qzero))$.
\item[A2)]
The order within  $Q_\sf$ 
extends the order between states of $\automata_\pf$:
if $q' < q$ in $\automata_\pf$ then $(\sf,q') < (\sf,q)$
and all other states of $Q_\sf$ are smaller.
A play eventually trapped in $\Q_\sf$ is \emph{$\automata_\pf$-accepting} if
its projection on $R_\pf$ is an accepting computation of $\automata_\pf$.
\item[A3)]
Every play eventually trapped in $Q_\sf$ with 
$\sf\in\{\exists \pf, \forall \pf \}$  is $\automata_\pf$-accepting.
\item[A4)]
Almost-every play eventually trapped in $Q_\sf$ with 
$\sf\in\{\proba_{>0}( \pf), \proba_{=1}( \pf) \}$ 
is $\automata_\pf$-accepting.
\item[A5)]
When the play enters a component $Q_\sf$ with 
$\sf= \proba_{>0}( \pf)$
then with positive probability
its continuation never enters neither $\ignore$ nor  $\bot$.
\end{itemize}

By design, these conditions can be expressed by $\Qall$, $\Qone$ and 
$\Qzero$ sets thanks to:
\begin{lemma}\label{lem:truth}
Assume a play $\pi$ is eventually trapped in $\Q_\sf$ with 
$\sf\in\{\exists \pf, \forall \pf,\proba_{>0}( \pf),\proba_{=1}( \pf) \}$.
Then $\pi$ is $\automata_\pf$-accepting if and only if
the second component of
its limsup  is an accepting state of $\automata_\pf$.

If moreover Eve plays truthfully then
$
(\text{$\pi$ is $\automata_\pf$-accepting})
\iff
(\pi \models \pf ) \enspace.
$
\end{lemma}

This automaton has \what\
because each time Adam asks for a proof,
the alternating automaton exits the current $\equiv$-class. It has size $O(2^{2^{|\pCTLFormula|}})$ as the union of a polynomial number of automata (the $\automata_\pf$) of size $O(2^{2^{|\pCTLFormula|}})$.

\medskip

We show that this automaton
recognizes exactly the set of models of the \ctlsa\ formula,
for that we describe a winning strategy for Eve if the input binary tree
is a model of the formula, and a winning strategy for Adam if not.

If the input tree $\lab:\{0,1\}^*\to\Sigma$ is a model of the formula,
then Eve has a winning strategy which maintains the following invariant:

\smallskip
\begin{description}
	\item[{\bf IE:}] 
	for every finite play $\pi$ whose last state is of the form
	$(s,\sf)$ with $s$ a node of the tree
	and $\sf$ a state formula then
	$
	\lab,s\models \sf\enspace.
	$
\end{description}
\smallskip

  First, Eve is always truthful about the valuations ${\bf b}$
 of the inner state formulas.
  Second, on vertices $(s,\sf_1\lor \sf_2)$
 Eve chooses a successor $(s,\sf_i)$ such that 
 $\mc,s \models \sf_i$,
 which is possible according to the invariant.
 
This guarantees the invariant {\bf IE} to be maintained,
since on vertices $(s,\sf_1\land \sf_2)$
controlled by Adam, the invariant guarantees 
$\mc,s \models \sf_1$ and $\mc,s \models \sf_2$
and when Adam asks for
a proof the invariant is maintained because Eve is truthful.
Moreover $\bot$ cannot be reached because Eve is truthful.

 To terminate the description of Eve strategy, 
 we specify 
the choice of Eve in a state $(\sf,q,E)$ 
when $\sf\in\{\exists \pf, \proba_{>0}(\pf)\}$.
When the play enters  $Q_\sf$,
Eve chooses a \emph{witness} of $\mc,s \models \sf$
(which holds according to {\bf IE}).
In case $\sf = \exists \pf$
the witness is a branch of the tree visiting the current node
and satisfying $\pf$.
In case $\sf = \proba_{>0}( \pf)$
the witness is a thick subtree whose root is the current node
and whose every branch
satisfies $\pf$
which exists according to Lemma~\ref{lem:thick}.
If $\sf=\exists \pf$ then Eve chooses the state
$(\sf,q')$ in the direction of the witness path and $\ignore$ in the other direction.
If $\sf= \proba_{>0}(\pf)$ then Eve chooses the state
$(\sf,q')$ in either direction staying in the witness subtree
and $\ignore$ in the other direction.

Now that Eve strategy is defined, we show that it is winning.

According to the invariant,
no play consistent with Eve strategy reaches $\sf=\bot$ 
thus according to Lemma~\ref{lem:truth}
 all plays have limsup in $\Qall$
and almost-all plays have limsup in $\Qone$.

When the finite play $\pi$ enters for the first time a component $Q_\sf$ with 
$\sf= \proba_{>0}( \pf)$ in a vertex $(s,\sf)$,
we show by induction on the topological structure of the automaton
that  there is $>0$ probability
that continuations of $\pi$ stay in $\Qzero$,
i.e. they do not enter the $\ignore$ state.
Let $T\subseteq \{0,1\}^*$ be the thick subtree chosen by 
Eve to witness $s \models  \proba_{>0}( \pf)$.
Every branch of $T$ is the projection
of a continuation of $\pi$,
and if this continuation
stays in $Q_\sf$ then by definition of Eve strategy,
it never enters $\ignore$.
Thus if all branches of $T$ are such projections,
we are done since $T$ is thick.
Otherwise,
there is at least one continuation of $\pi$
which leaves $Q_\sf$ in a state
$\neq \ignore$ and it will stay in non-$\ignore$ states
as long as it does not enter another $\proba_{>0}$-component,
thus we conclude with the inductive hypothesis.

\medskip

Conversely assume that the input binary tree $\lab:\{0,1\}^*\to\Sigma$ is not a model of the formula.
Then we describe a winning strategy for Adam which maintains the invariant:

\smallskip
\begin{description}
	\item[{\bf IA:}] 
	whenever the play reaches a vertex $(s,\sf)$ with $\sf\neq \top$ then
$\lab,s\not\models \sf\enspace$.
\end{description}
\smallskip

If Eve is not truthful when proposing a valuation
of the inner state formula,
Adam asks for a proof of one of the wrong entries,
which obviously maintains the invariant {\bf IA}.
And from every state $(s,\sf_1\land \sf_2)$,
Adam moves to either successor $(s,\sf_i)$ which maintains the invariant {\bf IA}.

Fix some strategy $\sigma$ for Eve.
If any play reaches $\bot$ then this falsifies the $\Qall$-condition and Adam wins.

We show that there exists at least one play eventually trapped
in some component $Q_\sf$
with
$\sf\in\{\exists \pf, \forall \pf,\proba_{>0}( \pf),\proba_{=1}( \pf) \}$.
By design of the transitions,
if a play enters $\ignore$
then a play of the same length does not.
And according to the invariant {\bf IA},
the only way to enter $\top$
is from a state of the form $(s,(\exists \pf,q,E))$
and in this case there is a play of the same length which does not enter $\top$.

Let $\pi$ be the finite play corresponding to the moment
the play enters $Q_\sf$, in some vertex $(s,\sf)$.
We can choose $\sf$ minimal which implies that Eve is truthful
in every continuation of $\pi$ whose last state is in $Q_\sf$.
According to Lemma~\ref{lem:truth},
for every continuation $\pi'$ of $\pi$ which stays
in $Q_\sf$,
$
(\text{$\pi'$ is $\automata_\pf$-accepting})
\iff
(\pi' \models \pf ) \enspace.
$
According to the invariant {\bf IA}
$\lab,s\not\models \sf$ thus
in all cases Adam wins:
\begin{itemize}
\item
if $\sf=\forall \pf$ or $\sf=\exists \pf$
 one of the continuations of $\pi$ stays in $Q_\sf$
but is not $\automata_\pf$-accepting.
\item
If $\sf=\proba_{=1}(\pf)$
there is $>0$-probability that a continuation of $\pi$
stays in $Q_\sf$ 
and is not $\automata_\pf$-accepting.
\item
If $\sf=\proba_{>0}( \pf)$
then there is probability $0$ that the play stays in
$Q_\sf$ and is $\automata_\pf$-accepting.
Thus either almost-every continuation of $\pi$
enters $\ignore\not\in\Qzero$
or there is 
$>0$-probability that a continuation of $\pi$
stays in $Q_\sf$ 
and is not $\automata_\pf$-accepting.
\end{itemize}

Thus  Lemma~\ref{lem:pctltobc} is proved
when the input formula $\pCTLFormula$ is in \ctlsa.

\paragraph*{Optimizing the construction for variants and fragments.}

If $\pCTLFormula$ is an \ECTLa\ formula,
as already noticed 
the deterministic parity automaton $\automata_\pf$ is already included in $\pCTLFormula$, thus its size is linear in 
$\mid \pCTLFormula\mid$.
As a consequence, the size of the alternating automaton is "only" exponential in the size of $\pCTLFormula$.

If $\pCTLFormula$  belongs to the fragment \ctl\allop, every state subformula $\pf$ of $\pCTLFormula$
has a single subformula
which is a path formula of type
$\next \pf' \mid \pf_1 \until \pf_2 \mid \always\pf'$.
Thus the subformula valuations proposed by Eve to Adam consist in one or two bits, thus they have constant size instead of exponential size.
Moreover, the deterministic parity automaton $\automata_\pf$ has  at most two states:
for $\next\pf'$ and $\always\pf'$ this is a one-state B\"uchi automaton,
for  $\pf_1 \until \pf_2$ this is a two-state co-B\"uchi automaton.
Finally the state space of the alternating automaton is linear in the size of the input formula.

When
the formula has no deterministic quantifier,
i.e. when it belongs to the fragment \ctls\pop\
then condition A3 becomes trivial. Since the  $\Qall$ condition is 
not used in any of the other conditions A1-A5, the alternating automaton is $\Qall$-trivial.
\end{proof}

\end{document}